\documentclass[a4paper,UKenglish,cleveref, autoref, thm-restate]{lipics-v2021}

\title{Reachability in Fixed-Dimensional Continuous VASS} 

\titlerunning{Reachability in Fixed-Dimensional Continuous VASS} 

\author{Michal Ajdarów}{University of Liverpool, United Kingdom}{Michal.Ajdarow@liverpool.ac.uk	}{https://orcid.org/0000-0003-0694-0944}{Supported by the EPSRC project EP/X042596/1}
\author{A. R. Balasubramanian}{Max Planck Institute for Software Systems (MPI-SWS), Kaiserslautern, Germany \and \url{arbalan96@github.io}}{bayikudi@mpi-sws.org}{https://orcid.org/0000-0002-7258-5445}{This research was sponsored in part by the Deutsche Forschungsgemeinschaft project \href{https://gepris.dfg.de/project/389792660}{389792660} TRR 248–CPEC}

\author{Łukasz Orlikowski}{University of Warsaw, Poland}{l.orlikowski@mimuw.edu.pl}{https://orcid.org/0009-0001-4727-2068}{Supported by the ERC grant INFSYS, agreement no. 950398}

\authorrunning{M. Ajdarów, A. R. Balasubramanian and Ł. Orlikowski} 

\Copyright{Michal Ajdarów, A. R. Balasubramanian and Łukasz Orlikowski} 



\category{} 

\relatedversion{} 

\nolinenumbers 

\EventEditors{Ana Sokolova and Patrick Totzke}
\EventNoEds{2}
\EventLongTitle{37th International Conference on Concurrency Theory (CONCUR 2026)}
\EventShortTitle{CONCUR 2026}
\EventAcronym{CONCUR}
\EventYear{2026}
\EventDate{September 1--4, 2026}
\EventLocation{Liverpool, UK}
\EventLogo{}
\SeriesVolume{391}
\ArticleNo{41}

\usepackage{todonotes}
\usepackage{environ}
\usepackage{apxproof}
\input{macros}

\begin{document}

\maketitle

\begin{abstract}
Vector Addition System with States (VASS) are a ubiquitous model of infinite-state systems consisting of a set of non-negative counters which can be incremented and decremented. It is known that the reachability problem for VASS is Ackermann-complete. Because of this huge complexity, various over-approximations of VASS have been studied in the literature. One such over-approximation is continuous VASS (CVASS), in which the counters are (non-negative) rational numbers and whenever a vector is added to the current counter values, it is first scaled with an arbitrarily chosen rational factor between zero and one. It is known that the reachability problem for CVASS is $\NP$-complete.

In this paper, we initiate the study of fixed-dimensional CVASS, i.e., CVASS with a fixed number of counters. We study both the reachability and coverability problems, under both unary and binary encodings as well as over both the non-negative and the rational semantics. This gives rise to a collection of eight different problems. As our main result, we prove a complexity dichotomy for all of these eight problems when the transition vectors are over the rationals: For dimension 1, all of the eight problems are in $\AC^1$, and so within $\mathsf{P}$, whereas for any dimension at least 2, all of the eight problems are $\NP$-complete. Furthermore, the hardness holds even when the underlying automaton is acyclic. This presents the first $\NP$-hardness result for acyclic unary counter systems over fixed dimensions. To achieve this result, we present a new technique called the Egyptian prime fractions technique.

Finally, we also study these problems when the transition vectors are over the integers. Except for dimension 2, we classify the complexity of these problems over the non-negative semantics: For dimension 1, all of the problems are in $\AC^1$, whereas for dimensions 3 and above, all of the problems are $\NP$-complete.
\end{abstract}
	
\section{Introduction}

Vector Addition System with States (VASS) (or equivalently Petri nets) are one of the most well-studied models of infinite-state systems with applications in logic, automata, verification and modeling concurrent and business process~\cite{siglog/Schmitz16,EsparzaGLM17,GermanS92,Aalst98}. A VASS consists of finitely many states along with finitely many counters ranging over the natural numbers. The transitions of a VASS allow it to move from one state to another, whilst incremeting or decrementing the values of the counters (provided the new values are also natural numbers). One of the most important problems studied for VASS is the \emph{reachability} problem, which asks if a given source configuration $(p,\bu)$ can reach a given target configuration $(q,\bv)$. After decades of research, the complexity of the reachability problem for VASS was finally settled to be Ackermann-complete~\cite{Leroux21,CzerwinskiO21}.

Due to this huge complexity of the reachability problem, various \emph{over-approximations} of the reachability relation for VASS have been studied in the literature (See~\cite{Blondin20} for a survey). Since usually the target configuration in the reachability problem is an error configuration that we do not want the source configuration to reach, if we show that this error is not reachable in an over-approximation, then the same applies for the original VASS as well. If the over-approximation is tractable, then this provides a method to prove safety that trades completeness for efficient analysability.

One particular over-approximation that has proven to be quite useful in recent years is the \emph{continuous semantics} (equivalently also called as continuous VASS or CVASS)~\cite{david2005discrete,BlondinH17}. In this semantics, the transitions of a VASS are allowed to be fired \emph{fractionally}. This means that, to fire a transition $t$, we first pick a fraction $\alpha \in (0,1]$ and then execute the updates of $t$ by the fraction $\alpha$, i.e.,
if the transition $t$ originally incremented (resp. decremented) some counter by a value $c$, it will now be incremented (resp. decremented) by the value $\alpha \cdot c$. (Hence picking $\alpha = 1$ is the same as performing the original transition). Because of this \emph{fractional firing}, counter values can now be (non-negative) rational numbers.
Depending on whether we allow the counters to only be non-negative or allow them to go below zero, we get two types of continuous semantics, namely the $\qnz$-semantics and $\qn$-semantics respectively. 

It turns out that both types of continuous semantics are remarkably well-behaved in terms of complexity. Reachability for both of them is $\NP$-complete and is hence much lower than the usual semantics of VASS~\cite{BlondinH17}. Besides this tractability result, the continuous semantics is quite useful in analysing practical benchmarks for problems related to reachability~\cite{EsparzaLMMN14,BlondinFHH17}, has applications in modeling client-server systems~\cite{client-server}, biological networks and biochemical settings~\cite{HaarK25,HeinerGD08,HerajyH18} and in parameterized verification~\cite{BalasubramanianER23,Balasubramanian20}.
However, despite all of this wealth of positive results on continuous VASS, not much research has been devoted to studying them over fixed-dimensions, i.e., studying continuous VASS over a fixed number of counters. This situation is especially surprising, considering the fact that for VASS and many of its over-approximations, better complexity results are known for small dimensions in different settings~\cite{BlondinEFGHLMT21,CzerwinskiJ0O25,GurariI81,BaumannDGIMSZ23}.

In this paper, we initiate the systematic study of the reachability as well as the coverability problem in fixed-dimensional continuous VASS. In the latter problem,
we are given two configurations $(p,\bu)$ and $(q,\bv)$ and we want to check 
if $(p,\bu)$ can reach a configuration of the form $(q,\bv')$ for some $\bv' \ge \bv$.
We study both of these problems for both the $\qnz$-semantics as well as the $\qn$-semantics, under both the unary and the binary encoding of numbers. (Hence, we get eight different problems, by picking one attribute in each of the three categories).

Our main result is a complete complexity classification of all of the eight problems for any fixed dimension, when the updates on the transitions are allowed to be rational numbers.
Namely, first we show that for dimension 1, all of the above eight problems are in $\AC^1$, and so within $\mathsf{P}$. Then, we show that 
for any dimension $d \ge 2$, all of the above eight problems are $\NP$-complete, and hence
as hard as the case of an arbitrary number of counters. To achieve the desired $\NP$-hardness result, we introduce a novel technique called the \emph{Egyptian prime fractions} technique, where we use \emph{sums of reciprocals of primes} to uniquely encode the assignments of a propositional formula into a single rational-valued counter.

Prior to our results, only an $\mathsf{NC}^2$ upper bound was known for 1-dimensional CVASS~\cite[Theorem 1]{BlondinLMOP23}, which we improve to an $\AC^1$ upper bound in this paper. Also, prior to our results, $\NP$-hardness was only known for reachability for 2-dimensional CVASS over binary encoding~\cite[Lemma 4.13]{BlondinH17}. 
We strengthen this result to apply even for coverability over the unary encoding. Furthermore, our $\NP$-hardness result, for all of the eight problems, already applies to \emph{acyclic} continuous VASS, i.e., even when the underlying control structure of the continuous VASS is acyclic. This presents the first $\NP$-hardness result for acyclic unary fixed-dimensional counter systems, as well as the first $\NP$-hardness result for 2-dimensional unary VASS  and its over-approximations. 

Finally, we also study these problems for the $\qnz$-semantics when the updates on the transitions are restricted to be integers. In this case, we present a complexity classification of these problems (reachability/coverability, unary/binary encoding) for any fixed dimension, except for the case of dimension 2: For dimension 1, all four problems are in $\AC^1$, whereas for dimensions 3 and above, all four problems are $\NP$-complete. Once again, the hardness results already hold for 3-dimensional acyclic unary continuous VASS. This is in stark contrast to VASS (and some of its other over-approximations and variants), where coverability over the unary encoding for any fixed dimension is in $\NL$~\cite{GurariI81,BaumannDGIMSZ23,Rackoff78}.

\section{Preliminaries}
\subparagraph*{Basic Notions. }
We use $\nn$, $\zz$, $\qq$ and \(\qq_+ \)  for the sets of non-negative integers, integers, rational numbers and non-negative rational numbers, respectively. For $a \ge 1$, let $[a] := \{1,2,\dots,a\}$. 
For a vector $\vr v \in \qq^d$, we denote its $j$-th coordinate by $\vr v[j]$. 
For $\vr u, \vr v \in \qq^d$, we write $\vr u \ge \vr v$ if $\vr u[i] \ge \vr v[i]$ for all $i \in [d]$.



\subparagraph*{Continuous Vector Addition System with States. }
Let $d \in \nn$. A $d$-dimensional continuous vector addition system with states ($d$-CVASS) is a tuple $\mach = (Q, T)$ where $Q$ is a finite set of control states and $T \subseteq Q \times \qq^d \times Q$ is a finite set of transitions. $\mach$ is called acyclic, if the directed graph induced by $\mach$ with $Q$ as vertices and edges corresponding to the transitions $T$, is an acyclic graph. 

Intuitively, the $d$-CVASS $\mach$ is a system with $d$ counters, each of them storing a rational value. Its transitions update the counters by incrementing or decrementing them in a specific manner, as described below.
A \emph{configuration} of $\mach$ is a pair  $p(\vr v)$ where $p \in Q$ is a state and $\vr v \in \qq^d$ is a vector representing the current values of the $d$ counters. 
Given two configurations $p(\vr u)$ and $q(\vr v)$, we say that there is a \emph{step} from $p(\vr u)$ to $q(\vr v)$ if there exists a transition $t = (p, \vr w, q)$ and a non-zero fraction $\alpha \in (0,1]$ such that $\vr v = \vr u + \alpha \cdot \vr w$. In this case,
we denote it by either $p(\vr u) \actqn{\alpha t} q(\vr v)$, or just as $p(\vr u) \actqn{} q(\vr v)$. Moreover, we use $p(\vr u) \actqnz{} q(\vr v)$ to denote that $p(\vr u) \actqn{} q(\vr v)$ and additionally both \(\vr u\geq 0\) and \(\vr v \geq 0 \). These step relations are called the $\qn$-semantics ($\actqn{}$) and the $\qnz$-semantics ($\actqnz{}$), respectively.

A \(\qq \)-run (also called a run over the $\qn$-semantics) is a sequence of configurations $q_0(\vr u_0), q_1(\vr u_1), \ldots, q_l(\vr u_l)$ such that $q_{i-1}(\vr u_{i-1}) \actqn{} q_i(\vr u_i)$ for every $i \in [l]$. The notion of a $\qnz$-run is defined analogously. We say that $p(\bu)$ can reach $q(\bv)$ under the $\qn$-semantics (resp. $\qnz$-semantics) if there exists a $\qn$-run (resp. $\qnz$-run) from $p(\bu)$ to $q(\bv)$. We use $\actqn{*}$ (resp $\actqnz{*}$) to denote the $\qn$-reachability relation (resp. $\qnz$-reachability relation) between configurations.
Finally, we say that $p(\bu)$ can cover $q(\bv)$ over the $\qn$-semantics (resp. $\qnz$-semantics) if there is a $\qn$-run (resp. $\qnz$-run) from 
$p(\bu)$ to some $q(\bv')$ with $\bv' \ge \bv$.


\subparagraph*{Reachability and Coverability}
The main focus of this paper are the reachability and coverability problems for CVASS in fixed dimension, where the number of counters $d$ is fixed and not part of the input.  We now define these problems formally.

\tcbset{
	problembox/.style={
		colback=gray!8,
		colframe=gray!50,
		boxrule=0.4pt,
		arc=5pt,
		left=5pt, right=5pt, top=3pt, bottom=0pt,
		width=0.47\textwidth,
		equal height group=problems,
		nobeforeafter
	}
}
\noindent
\begin{tcolorbox}[problembox]
	\begin{statement}
		\begin{center}$d$-$\Reach_{\qn}$\end{center}
		\vspace{-28pt}
		\probleminput{A $d$-CVASS and two of its configurations $p(\bu), q(\bv)$.}
		\problemquestion{Can $p(\bu)$ reach $q(\bv)$ over the $\qn$-semantics?}
	\end{statement}\vspace{-7pt}
\end{tcolorbox}
\hfill
\begin{tcolorbox}[problembox]
	\begin{statement}
		\begin{center}$d$-$\Cover_{\qn}$\end{center}
		\vspace{-28pt}
		\probleminput{A $d$-CVASS and two of its configurations $p(\bu), q(\bv)$.}
		\problemquestion{Can $p(\bu)$ cover $q(\bv)$ over the $\qn$-semantics?}
	\end{statement}\vspace{-7pt}
\end{tcolorbox}


The problems $d$-$\Reach_{\qnz}$ and $d$-$\Cover_{\qnz}$ are defined analogously, with $\actqn{*}$ replaced by $\actqnz{*}$. We further distinguish these problems by \emph{unary} and \emph{binary} encodings, depending on how the input is encoded. Thus, for each fixed $d$, we obtain eight variants, determined by reachability vs.\ coverability, $\qn$- vs.\ $\qnz$-semantics, and unary vs.\ binary encoding.


The main result of this paper is a complete complexity classification of these eight problems. This requires the (logspace-uniform) complexity class $\AC^1$, defined as the set of decision problems solvable by polynomial-sized unbounded fan-in Boolean circuits with depth $O(\log n)$.  It is known that  $\NL \subseteq \AC^1 \subseteq \mathsf{P}$~\cite[Fig. 4.6]{Vollmer}. Also, iterated integer multiplication, i.e., the task of multiplying a given set of $n$ integers, each encoded in binary using at most $n$ bits, is also known to be in $\AC^1$~\cite[Corollary 4.1]{HesseAB02}. 

We are now ready to state the two main results of the paper, which are as follows. (See the top part of Table~\ref{tab:cvass_complexity} for a succinct description of these results).

\begin{theorem}\label{thm:main-result-1-CVASS}
     1-$\Reach_{\qn}$, 1-$\Cover_{\qn}$, 1-$\Reach_{\qnz}$ and 1-$\Cover_{\qnz}$ are all in $\AC^1$ for both unary and binary encoding.
\end{theorem}

\begin{theorem}\label{thm:main-result-2-CVASS}
     For any $d \ge 2$, $d$-$\Reach_{\qn}$, $d$-$\Cover_{\qn}$, $d$-$\Reach_{\qnz}$ and $d$-$\Cover_{\qnz}$ are all $\NP$-complete for both unary and binary encoding.
\end{theorem}

In addition to this main result, we also consider a restricted version of CVASS where we only allow the transition vectors to range over the integers. Formally, we say that 
a $d$-CVASS $\mach = (Q,T)$ is a CVASS \emph{with integer updates} if $T \subseteq Q \times \zn^d \times Q$, i.e., all the vectors appearing in the transitions are over $\zn^d$. Our results in this regard are the following: (See the bottom part of Table~\ref{tab:cvass_complexity} for a succinct description of these results).

\begin{theorem}\label{thm:main-result-integer-updates-1-CVASS}
     For CVASS with integer updates, 1-$\Reach_{\qn}$, 1-$\Cover_{\qn}$, 1-$\Reach_{\qnz}$ and 1-$\Cover_{\qnz}$ are $\NL$-complete over the unary encoding and in $\AC^1$ over the binary encoding. 
\end{theorem}

\begin{theorem}\label{thm:main-result-integer-updates-2-CVASS} 
    For CVASS with integer updates, 2-$\Reach_{\qn}$, 2-$\Cover_{\qn}$, 2-$\Reach_{\qnz}$ and 2-$\Cover_{\qnz}$ are $\NP$-complete over the binary encoding.
\end{theorem}

\begin{theorem}\label{thm:main-result-integer-updates-3-CVASS}
    For CVASS with integer updates, for any $d \ge 3$, $d$-$\Reach_{\qnz}$ and $d$-$\Cover_{\qnz}$ are  $\NP$-complete for both unary and binary encoding.
\end{theorem}

We note that all of the $\NP$-hardness results hold even when the CVASS is acyclic. 





\begin{table}[htbp]
\resizebox{\textwidth}{!}{%
\begin{tabular}{lllllll}
\hline
\textbf{Problem} & \textbf{Semantics} & \textbf{Encoding} & \textbf{Dimension} &  \textbf{Transitions} & \textbf{Complexity} & \textbf{Theorem} \\
\hline
Coverability/Reachability & $\qnz$ / $\qn$ & Unary / Binary & 1 & Rational Updates & in $\AC^1$ & Theorem~\ref{thm:main-result-1-CVASS} \\
Coverability/Reachability & $\qnz$ / $\qn$ &  Unary / Binary & $\geq 2$ & Rational Updates & $\NP$-complete & Theorem~\ref{thm:main-result-2-CVASS} \\
\hline
Coverability/Reachability & $\qnz$ / $\qn$ & Unary & 1 & Integer Updates & $\NL$-complete & Theorem~\ref{thm:main-result-integer-updates-1-CVASS} \\
Coverability/Reachability & $\qnz$ / $\qn$ & Binary & 1 & Integer Updates & in $\AC^1$  & Theorem~\ref{thm:main-result-integer-updates-1-CVASS} \\
Coverability/Reachability & $\qnz$ / $\qn$ &  Binary & $\geq 2$ & Integer Updates & $\NP$-complete & Theorem~\ref{thm:main-result-integer-updates-2-CVASS} \\
Coverability/Reachability & $\qnz$ &  Unary / Binary & $\geq 3$ & Integer Updates & $\NP$-complete & Theorem~\ref{thm:main-result-integer-updates-3-CVASS}\\
\hline
\end{tabular}%
}
\captionsetup{justification=raggedright, singlelinecheck=false}
\caption{Complexity of reachability and coverability in fixed-dimensional CVASS; A row depicts all possible combinations of problems that can be obtained by the entries in that row; for instance the first row depicts all of the eight possible problems for dimension 1, and for all of them, the upper bound is $\AC^1$.}
\label{tab:cvass_complexity}
\end{table}

All of the $\NP$ membership results in the above theorems follow from the fact that reachability and coverability for $d$-CVASS (over both types of encoding and both semantics) is in $\NP$ for any $d$~\cite[Section IV.B]{BlondinH17}. Hence, the rest of this paper is dedicated to proving all of the other results of the theorems (or equivalently all of the other results in Table~\ref{tab:cvass_complexity}). Since there are a lot of problems that we need to consider, we make a remark that roughly reduces the number of problems by half.

\begin{remark}\label{rem:cov-to-reach}
    	We can easily reduce coverability to reachability in logarithmic space: Given a $d$-CVASS $\mach = (Q,T)$ with configurations $p(\bu)$ and $q(\bv)$, we modify $\mach$ by adding a new transition which moves from $q$ to a new state $q'$. At $q'$, we add $d$ self-loops, each decrementing one of the counters by 1. It is easy to see that
	(over any of the semantics) $p(\bu)$ can cover $q(\bv)$ in $\mach$ iff (over the corresponding semantics) $p(\bu)$ can reach $q'(\bv)$. 
\end{remark}

The rest of this paper is structured as follows: In Section~\ref{sec:1-CVASS}, we prove Theorems~\ref{thm:main-result-1-CVASS} and~\ref{thm:main-result-integer-updates-1-CVASS}. 
In Section~\ref{sec-2CVASS}, we prove Theorems~\ref{thm:main-result-2-CVASS} and~\ref{thm:main-result-integer-updates-2-CVASS}. Finally, in Section~\ref{sec-3CVASS},
we prove Theorem~\ref{thm:main-result-integer-updates-3-CVASS}.

\section{Upper Bounds for 1-CVASS}\label{sec:1-CVASS}

In this section, we prove Theorems~\ref{thm:main-result-1-CVASS} and~\ref{thm:main-result-integer-updates-1-CVASS}. To this end, we  prove the following main result.

\begin{theorem}\label{thm:1-CVASS}
    1-$\Reach_\qn$ and 1-$\Reach_{\qn_+}$ are both 
    \begin{itemize}
        \item $\NL$-complete for 1-CVASS with integer updates over the unary encoding
        \item and in $\AC^1$ for 1-CVASS over both unary and binary encoding.
    \end{itemize}
\end{theorem}

By Theorem~\ref{thm:1-CVASS} and Remark~\ref{rem:cov-to-reach}, Theorems~\ref{thm:main-result-1-CVASS}
and~\ref{thm:main-result-integer-updates-1-CVASS} immediately follow. (The $\NL$-hardness mentioned in Theorem~\ref{thm:main-result-integer-updates-1-CVASS} is inherited from graph reachability). So, the rest of this section is dedicated to proving Theorem~\ref{thm:1-CVASS}. For this, we first make a remark.

\begin{remark}\label{rem:source-target}
    Since we will be dealing with a single counter for the rest of this section,
    we will use variables like $x,y$ instead of their boldface counterparts $\bx,\by$ to denote counter values.
    
    Now, to prove the upper bounds for 1-CVASS, it suffices to give the algorithms in the case when the input instance is of the form $(\mach,p(x),q(y))$ with $x \le y$. Indeed, if we are given an instance where $x > y$,
    we can consider the \emph{reverse} CVASS $\mach^\dagger$ of $\mach$ which has the same states as $\mach$ with the transitions reversed: $(p',w,q')$ is a transition of $\mach$ iff $(q',-w,p')$ is a transition of $\mach^\dagger$. It is immediately seen that there is a $\qn$-run (resp. $\qnz$-run) from $p(x)$ to $q(y)$ in $\mach$ iff there is a $\qn$-run (resp. $\qnz$-run) from $q(y)$ to $p(x)$ in $\mach^\dagger$.
\end{remark}

We now move on to proving Theorem~\ref{thm:1-CVASS}. We do this in three stages: First,
we show that reachability for 1-CVASS over $\qnz$-semantics (and indeed over a slight generalization of the $\qnz$-semantics) 
can be reduced to $\qn$-semantics. In the second stage, we give a characterization of reachability in 1-CVASS over $\qn$-semantics. Then, in the final stage, we use this characterization to give the required upper bounds for 1-CVASS.
We note that some of the arguments in the first two stages are similar to some of the arguments in~\cite{BlondinLMOP23}, but our presentation makes it simpler to arrive at the upper bound. Furthermore, in the third stage, we improve upon the $\mathsf{NC}^2$ bound in~\cite{BlondinLMOP23} to an $\AC^1$ bound.

\subsection*{Stage 1: Reducing $\qnz$-semantics to $\qn$-semantics}

In the first stage, we shall show that reachability for 1-CVASS over the $\qnz$-semantics
reduces to reachability over the $\qn$-semantics. To this end, we consider a slight generalization of the $\qnz$-semantics so that it also covers an extension of 1-CVASS as considered by~\cite{BlondinLMOP23}.

Let $\inte$ be an interval over $\qn$, i.e., $\inte$ is of the form $(a,b), (a,b], [a,b)$ or $[a,b]$ for some $a \le b$ with $a, b \in \{-\infty,\infty\} \cup \qn$, with each of them interpreted the usual way. Given a 1-CVASS $\mach = (Q,T)$, the $\inte$-semantics of $\mach$ is the restriction of the $\qn$-semantics to only configurations of the form $p(x)$ with $p \in Q$ and $x \in \inte$, i.e.,
we say that $p(x) \act{}_\inte q(y)$ if $p(x) \actqn{} q(y)$ and $x, y \in \inte$.
Similar to $\qn$-semantics, we can now define runs over the $\inte$-semantics and the reachability problem 1-$\Reach_\inte$ for it.
Note that the $\qnz$-semantics is simply the $\inte$-semantics where $\inte = [0,\infty)$.
As our main result of this stage, we show that

\begin{restatable}{theorem}{intetoqn}\label{thm:inte-to-qn}
	For any interval $\inte$, 1-$\Reach_\inte$ can be reduced in logarithmic space to
	1-$\Reach_\qn$. Furthermore the reduction preserves the type of encoding (unary or binary) and the domain of the update vectors (rationals or integers).
\end{restatable}

\begin{proof}[Proof Sketch]
    Let $\mach = (Q,T)$ be an instance of reachability in 1-CVASS with two configurations $p(x)$ and $q(y)$ and let $\inte$ be an interval. The same argument mentioned in Remark~\ref{rem:source-target} allows us to assume that $x \le y$. We shall now construct another 1-CVASS $\mach'$ with two configurations $p'(x')$ and $q'(y')$ such that $p(x) \act{*}_{\inte} q(y)$ in $\mach$ iff $p'(x') \actqn{*} q'(y')$ in $\mach'$.
    
    First note that if at least one of $x,y \notin \inte$, then there can be no $\inte$-run from $p(x)$ to $q(y)$ in $\mach$. In that case, we can simply take $(\mach',p'(x'),q'(y'))$ to be some trivial no-instance of 1-$\Reach_\qn$. Hence, we can assume that $x, y \in \inte$. The required construction of $\mach'$ is now obtained by means of a case analysis on $x,y$. To this end, we say that $x$ (resp. $y$) is an extreme point of $\inte$ if $\inte = [x,b)$ or $\inte = [x,b]$ for some $b$ (resp. if $\inte = (a,y]$ or $\inte = [a,y]$ for some $a$). Our case analysis  considers all four possibilities.
    
    \textbf{Case 1: $x,y$ are not extreme points of $\inte$. } 
    Therefore, there exists an $\epsilon > 0$ such that $[x-\epsilon,y+\epsilon] \subseteq \inte$. In this case, we show that $p(x) \act{*}_\inte q(y)$ is possible in $\mach$ iff
 	$p(x) \actqn{*} q(y)$ is possible in $\mach$, which then trivially gives the desired reduction.
 	The left-to-right implication is immediate. For the other side, suppose $p(x) := p_0(x_0) \actqn{\alpha_1 t_1} p_1(x_1) \actqn{\alpha_2 t_2} p_2(x_2) \dots \actqn{\alpha_n t_n} p_n(x_n) := q(y)$ is  a run over the $\qn$-semantics.
 	We say that the transition $t_i := (p_{i-1},w_i,p_i)$ is positive, negative or neutral, if $w_i > 0, w_i < 0$ or $w_i = 0$, respectively.
 	
 	Now, if all the intermediate counter values are already in $[x-\epsilon/2,y+\epsilon/2] \subseteq \inte$, then we are done. Otherwise, (since $x_0 = x_0$ and $x_n = y$), let $i$ be the first position and $j$ be the last position such that $x_{i+1}, x_{j-1} \notin [x-\epsilon/2,y+\epsilon/2]$. Hence, 
 	all the counter values up till $x_i$ and all the counter values from $x_j$ onwards belong to $[x-\epsilon/2,y+\epsilon/2]$. We now give another run between $p_i(x_i)$ and $p_j(x_j)$ so that all the intermediate counter values remain in $[x-\epsilon,y+\epsilon]$. Since $[x-\epsilon,y+\epsilon] \subseteq \inte$, plugging this new run between $p_i(x_i)$ and $p_j(x_j)$ into the old run will then give a $\inte$-run from $p(x)$ to $q(y)$, thereby completing the proof of this case.
 	
 	We accomplish this by further considering three different subcases:
    The first subcase is when $x_i = x_j$. Here, it suffices to simply fire
    each positive and negative transition in $\{t_{i+1},\dots,t_j\}$ with a sufficiently small fraction so that
    each of their combined effects is equal and at most $\epsilon/2$, which
    will ensure that we always stay in the range $[x_i-\epsilon/2,x_j+\epsilon/2] \subseteq [x-\epsilon,y+\epsilon]$. The second subcase is when $x_i < x_j$.
    Hence, intuitively, negative transitions can be suppressed as much as needed. So, we fire each negative transition in $\{t_{i+1},\dots,t_j\}$
    with a sufficiently small fraction so that their combined effect (say $E_-$) is at most $\epsilon/2$. Then we scale each positive transition with a suitable fraction
    so that its combined effect equals \emph{exactly} $x_j - x_i + E_-$. This
    will ensure that at any point in the run, the counter values stay in the range 
    $[x_i-\epsilon/2,x_j+\epsilon/2] \subseteq [x-\epsilon,y+\epsilon]$. Furthermore,
    by construction, the combined effect of the positive and negative transitions
    is exactly $x_j - x_i$ and so we reach counter value $x_j$ at the end.
    Finally, the case of $x_i > x_j$ is dual to the previous case, with the roles of positive and negative transitions swapped. 
    
    \textbf{Case 2: $x$ is an extreme point, $y$ is not.} Therefore, $\inte$ is of the form $[x,b)$ or $[x,b]$
	with $y < b$. This means that in any $\inte$-run from $p(x)$ and $q(y)$ in $\mach$, the first non-neutral transition must be a positive transition, as otherwise the counter value would have to leave the interval $\inte$. 
	Hence, any $\inte$-run from $p(x)$ to $q(y)$ in $\mach$ must necessarily be of the form $p(x) \act{*}_\qn d(x) \act{}_\qn e(z) \act{*}_\inte q(y)$ where $z > x$ and the run between $p(x)$ and $d(x)$ consists only of neutral transitions. 
	Furthermore, since $z > x$, it follows that there is an $\epsilon > 0 $ such that $[z-\epsilon,y+\epsilon] \subseteq \inte$. Applying the previous case to the run
	$e(z) \act{*}_\inte q(y)$, we can conclude that this can happen iff $e(z) \actqn{*} q(y)$. Hence, $p(x) \act{*}_\inte q(y)$ is possible iff there is a run of the form $p(x) \act{*}_\qn d(x) \act{}_\qn e(z) \act{*}_\qn q(y)$ where $z > x$ and the run between $p(x)$ and $d(x)$ consists only of neutral transitions.
	We can now easily construct another 1-CVASS $\mach'$ that explicitly tracks $\qn$-runs of this form, by ensuring that the first non-neutral transition that is fired is a positive transition. This can be done by allowing for only neutral transitions until a positive transition is fired, after which any transition is allowed. 
	
    \textbf{Case 3: $x$ is not an extreme point, $y$ is. } This case is a dual to Case 2, where the last non-neutral transition fired must be a positive transition. Hence, we now have to track the last non-neutral transition that is fired,
    which can be done in a manner dual to Case 2.
    
    \textbf{Case 4: $x,y$ are both extreme points. } Therefore, $\inte = [x,y]$. If $x = y$, we modify $\mach$ so that it has only neutral transitions. Otherwise,  we have to track both the first and last fired non-neutral transitions. So we combine the constructions from the above two cases.
\end{proof}

The full proof can be found in Appendix~\ref{app-subsec:inte-to-qn}. We now move on to the second stage.

\subsection*{Stage 2: Characterizing Reachability over $\qn$-semantics.}

In this stage, we give a precise characterization of reachability over the $\qn$-semantics. To state our characterization, we need to set up some notation.

Let $\mach = (Q,T)$ be a 1-CVASS. First, notice that $\mach$ naturally defines a finite labelled graph $\graph_\mach$ whose vertices are $Q$ and whose edges are given by the transition relation $T$, i.e., any transition $(p,w,q)$ corresponds to a labelled edge between $p$ and $q$ with label $w$.
Given a path $P$ in this graph $p_0 \act{w_1} p_1 \act{w_2} \dots p_{n-1} \act{w_n} p_n$ we define the sets $T_{+}(P), T_0(P)$ and $T_-(P)$ as $ \{i : w_i > 0\}, \{i : w_i = 0\}$ and $\{i : w_i < 0\}$, respectively.

Now, let $p(x)$ and $q(y)$ be two configurations of $\mach$. By Remark~\ref{rem:source-target}, we can assume that $x \le y$. We now show how to decide if $p(x)$ can reach $q(y)$, first in the case when $x = y$ and then in the case when $x < y$.

\subparagraph*{Case 1: $x = y$. } In this case, the characterization for reachability is simple and is as follows.

\begin{restatable}{theorem}{characequality}\label{thm:charac-equality}
	There is a $\qn$-run from $p(x)$ to $q(x)$ in $\mach$ iff there is a path $P$ in $\graph_\mach$ from $p$ to $q$
	such that either $T_+(P), T_-(P) \neq \emptyset$ or $T_+(P) = T_-(P) = \emptyset$.
\end{restatable}

    

The proof of this theorem can be found in Appendix~\ref{app-subsec:charac-equality}. We now move to the next case.

\subparagraph*{Case 2: $x < y$.} To state the characterization for this case, we need some notation and a couple of preparatory intermediate results. For a path $P$ in $\graph_\mach$ given by $p_0 \act{w_1} p_1 \act{w_2} \cdots p_{n-1} \act{w_n} p_n$, we let $W(P) = \sum_{1 \le i \le n} w_i$. Having introduced this notation, we prove the following proposition, which gives a simple sufficient condition for reachability.

\begin{proposition}\label{prop:simple-path}
    Suppose $P$ is a path from $d$ to $e$ in $\graph_\mach$ such that $W(P) \neq 0$, $W(P) \ge v-u$ for some $u, v$ and $W(P)$ has the same sign as $v - u$, i.e., $W(P) < 0 \iff v-u < 0$. Then there is a $\qn$-run from $d(u)$ to $e(v)$ in $\mach$.
\end{proposition}

\begin{proof}
    Let $P$ be a path with $W(P) \ge v-u$ given by $d := p_0 \act{w_1} p_1 \act{w_2} \dots p_{n-1} \act{w_n} p_n := e$.
    Let $t_i := (p_{i-1},w_i,p_i)$ be the transition in $\mach$ corresponding to the $i^{th}$ edge of $P$.
    Then, firing each $t_i$ with fraction $\frac{v-u}{W(P)}$ is easily seen to be a $\qn$-run from $d(u)$ to $e(v)$ in $\mach$.
\end{proof}

The next preparatory result that we will need is regarding cycles in $\graph_\mach$. To this end, a cycle $C$ of $\graph_\mach$ is said to be \emph{pumpable} if at least one of the labels in $C$ is strictly positive, i.e., $T_+(C) \neq \emptyset$. As the name suggests, we can use pumpable cycles to increase the counter value.

\begin{restatable}{proposition}{pumpablecycle}\label{prop:pumpable-cycle}
	Suppose $C$ is a pumpable cycle from $r$ to $r$, for some state $r$. Then, for any $u, v$ with $u < v$,
	we have a $\qn$-run from $r(u)$ to $r(v)$.
\end{restatable}

\begin{proof}[Proof Sketch]
	Let $C$ be a pumpable cycle given by $r := r_0 \act{w_1} r_1 \act{w_2} \dots r_{n-1} \act{w_n} r_n := r$. 
	Let $t_i := (r_{i-1},w_i,r_i)$ be the transitions in $\mach$ corresponding to this path.	
	We will fire each
	negative transition corresponding to $T_-(C)$ with a very tiny fraction so that their combined effect is of the form $-\epsilon/m$ for a small $\epsilon$ and large $m$. We then fire each positive transition corresponding to $T_+(C)$ with an appropriately chosen fraction so that their
	combined effect is exactly equal to $(v-u+\epsilon)/m$. 
	This enables us to construct a run from $r(u)$ to $r(u+\Delta)$ where $\Delta = \frac{y-x}{m}$. 
	Repeating this same run over and over, we can reach 
	$r(u+2\Delta), r(u+3\Delta)$ and so on till $r(u+m\Delta) := r(v)$.
\end{proof}

The full proof of this proposition can be found in Appendix~\ref{app-subsec:pumpable-cycle}. Now, with the above two propositions in hand, we are almost ready to state our characterization of reachability over $\qn$-semantics.
The final definition we need is the following: We say that a triple $(I,C,F)$ is a \emph{knot} of $\graph_\mach$ from $i$ to $f$, if there is a state $c$ such that $I$ is a path from $i$ to $c$, $C$ is a pumpable cycle from $c$ to $c$
and $F$ is a path from $c$ to $f$.
We are now ready to state our characterization of reachability.

\begin{restatable}{theorem}{characinequality}\label{thm:charac-inequality}
	For any $x < y$, there is a $\qn$-run from $p(x)$ to $q(y)$ in $\mach$ iff either there is a knot $K$ from $p$ to $q$ in $\graph_\mach$ or there is a path $P$ of length at most $|Q|$ from $p$ to $q$ in $\graph_\mach$ such that $W(P) \ge y-x$.
\end{restatable}

For the left-to-right direction, given any $\qn$-run from $p(x)$ to $q(y)$, either
it contains a cycle with a positive transition, and so we can extract a knot from it,
or it only contains cycles with no positive transitions, each of which can be safely deleted
to get a path $P$ which must satisfy $W(P) \ge y-x$. For the other direction,
we use Propositions~\ref{prop:simple-path} and~\ref{prop:pumpable-cycle}.
The full proof can be found in Appendix~\ref{app-subsec:charac-inequality}.

Finally, by Remark~\ref{rem:source-target}, we do not have to consider separately the case when $x > y$. Hence, this allows us to complete the second stage and move on to the third stage.

\subsection*{Stage 3: Upper Bounds for Reachability in 1-CVASS}

In this final stage, we use the characterization from the previous stage to give upper bounds for $\qn$-reachability in 1-CVASS. By the characterizations given in Theorem~\ref{thm:charac-equality} and Theorem~\ref{thm:charac-inequality}, the existence of a $\qn$-run between $p(x)$ and $q(y)$ in a 1-CVASS $\mach$
with $x \le y$ is governed by one of the following three conditions:
\begin{itemize}
    \item Condition 1: $x = y$ and there is a path $P$ in $\graph_\mach$ from $p$ to $q$ such that either $T_+(P), T_-(P) \neq \emptyset$ or $T_+(P) = T_-(P) = \emptyset$.
    \item Condition 2:  $x < y$ and there is a knot $K$ in $\graph_\mach$ from $p$ to $q$.
    \item Condition 3: $x < y$ and there is a path $P$ of length at most $|Q|$ in $\graph_\mach$ from $p$ to $q$ such that $W(P) \ge y-x$.
\end{itemize}

We prove that
\begin{restatable}{lemma}{testingconditions}\label{lem:testing-conditions}
    Conditions 1 and 2 can be decided in $\NL$. Condition 3 can be decided in $\NL$ if $\mach$ is a 1-CVASS with integer updates over unary encoding and in $\AC^1$ otherwise. 
\end{restatable}

Indeed, the first two conditions can be easily seen to be decidable in $\NL$. Indeed, for the first condition, if there is such a path, there has to be one of length at most $3|Q|$. Similarly, we can show that for the second condition, if there is such a knot, there has to be one of length at most $4|Q|$. These two characterizations immediately place the problem in $\NL$, by allowing for a guess-and-check algorithm on the graph $\graph_\mach$. Similarly, the third condition
can be decided in $\NL$ if $\mach$ is a 1-CVASS with integer updates and unary encoding.

If $\mach$ is a 1-CVASS with integer updates and binary encoding, we note that
the problem of checking for a path $P$ in $\graph_\mach$ of length at most $|Q|$ with $W(P) \ge y-x$ is equivalent to an instance of \emph{coverability in 1-dimensional integer VASS}, i.e., 1-CVASS over the $\qn$-semantics where all the fractions have to be 1. This latter problem is known to be in $\AC^1$~\cite[Theorem III.2]{ShakibaSZ25}. Finally, if $\mach$ has rational updates, we can first multiply the labels of $\graph_\mach$ with the product $B$ of all of its denominators, as well as the denominators of $x,y$. In this new graph (which only has integer labels) we have to find a path $P$ with $W(P) \ge B(y-x)$.
Hence, once again the problem is equivalent to an instance of coverability in 1-dimensional integer VASS. Since the product of a given collection of numbers
can be computed in $\AC^1$~\cite[Corollary 4.1]{HesseAB02}, this proves that the third condition can be tested in $\AC^1$
when $\mach$ has rational updates as well. A full proof of Lemma~\ref{lem:testing-conditions} conditions can be found in Appendix~\ref{app-subsec:testing-conditions}.

Finally, since $\NL \subseteq \AC^1$, by Remark~\ref{rem:source-target} and by Theorem~\ref{thm:inte-to-qn}, this then proves Theorem~\ref{thm:1-CVASS}, which completes the main result of this section.

\begin{remark}
    Our upper bound can also be extended to work for 1-CVASS \emph{with zero-tests}. In this extension, a subset of states $Z$ of the 1-CVASS is designated as the set of zero-test states. A run is then considered valid iff every configuration of the form $q(x)$ appearing in it with $q \in Z$ satisfies $x = 0$. In~\cite[Section 3.5]{BlondinLMOP23}, it is shown how to reduce the reachability problem for this extension to reachability over 1-CVASS. That reduction, along with our improved upper bound for 1-CVASS,
    allows us to conclude that reachability for 1-CVASS with zero-tests is also in $\AC^1$.
\end{remark}

\section{Lower Bounds for 2-CVASS}
\label{sec-2CVASS}
In this section, we prove Theorems~\ref{thm:main-result-2-CVASS} and~\ref{thm:main-result-integer-updates-2-CVASS}. To this end, we prove the following main result.

\begin{theorem}\label{thm:2-CVASS}
	2-$\Cover_\qn$ and 2-$\Cover_{\qnz}$ are both $\NP$-hard for acyclic 2-CVASS 
	\begin{itemize}
		\item with integer updates encoded in binary;
		\item with rational updates encoded in unary.
	\end{itemize}
\end{theorem}



Let us first see how this result implies Theorems~\ref{thm:main-result-2-CVASS} and~\ref{thm:main-result-integer-updates-2-CVASS}. First, by Remark~\ref{rem:cov-to-reach}, all of the $\NP$-hardness results in the statements of these two theorems follow from Theorem~\ref{thm:2-CVASS}. Finally, all of the $\NP$ membership results in the statements of these two theorems follow from the fact that reachability for $d$-CVASS (over both types of encoding and both semantics) is in $\NP$
for any $d$~\cite[Section IV.B]{BlondinH17}. So, the rest of this section is dedicated to proving Theorem~\ref{thm:2-CVASS}. We begin by showing that the first part of Theorem~\ref{thm:2-CVASS} follows from the second part.


		
		 Let $\mach = (Q, \transitions)$ be an acyclic unary $2$-CVASS and \(p(\bu) \), \(q(\bv) \) be some configurations of \(\mach \). Let $B$ be the product of all the denominators of the numbers appearing in $T$ as well as in $\bu,\bv$. 
		 Let $\mach'$ be the acyclic \emph{binary} 2-CVASS that we obtain by multiplying all
		 of the transition vectors of $\mach$ by $B$. It is immediately clear that
		 $p(\bu)$ can cover $q(\bv)$ in $\mach$ (over either of the semantics) iff 
		 $p(B\bu)$ can cover $q(B\bv)$ in $\mach'$ (over the corresponding semantics).
		 Since it is easily seen that $B$ can be computed in binary in polynomial time,
		 this proves that the first part of Theorem~\ref{thm:2-CVASS} follows from the second part.

 Hence it suffices to prove the second part of Theorem~\ref{thm:2-CVASS}, which we do in two stages. In the first stage, we consider a restricted version of the reachability problem and show that this restricted reachability problem is $\NP$-hard already for acyclic unary 1-CVASS (over both the semantics). In the second stage, we then reduce restricted reachability to coverability by the addition of an extra counter.


\subsection*{Stage 1: Hardness of Restricted Reachability of Acyclic Unary 1-CVASS }

In this stage, we show $\NP$-hardness of a restricted version of the reachability problem (over both the semantics). To define this problem formally, we set up some notation. 
A restricted $\qn$-run of a CVASS $\mach$ is a $\qn$-run in which all of the fractions appearing in each step are exactly 1. More precisely, a restricted $\qn$-run is any run of the form $q_0(\bu_0) \actqn{\alpha_1 t_1 } q_1(\bu_1) \actqn{\alpha_2 t_2} q_2(\bu_2) \cdots \actqn{\alpha_n t_n} q_n(\bu_n)$ where each $\alpha_i = 1$.
In the restricted $\qn$-reachability problem, we are given a CVASS $\mach$ and two configurations $p(\bx)$ and $q(\by)$ and we have to decide if $p(\bx)$ can reach $q(\by)$ by a restricted $\qn$-run. Similarly, we can define restricted $\qnz$-runs and restricted $\qnz$-reachability. The main result of this stage is to show that

\begin{theorem}\label{thm-restrictedreach-onecvass-unary-np-hard}
	Restricted $\qn$-reachability and restricted $\qnz$-reachability is $\NP$-hard for acyclic unary \(1\)-CVASS.
\end{theorem}

Theorem~\ref{thm-restrictedreach-onecvass-unary-np-hard}  is surprising because, by restricting the fractions to always be 1, the model behaves similar to a 1-dimensional unary (integer) VASS, for which reachability is in $\NL$~\cite{GurariI81,BlondinEFGHLMT21}. However, unlike those machines, the  difference here is that transition updates in a 1-CVASS are allowed to have rational values, which as we will see allows for much more power.

To prove the $\NP$-hardness, we give a reduction from a constrained version of 3-SAT. More specifically, we consider as input 3-CNF formulas where each variable appears at most 4 times, and we want to decide if such a formula is satisfiable. This problem is  $\NP$-complete~\cite{tovey}.

Hence, let $\varphi$ be a given 3-CNF formula over $n$ variables $x_1,\dots,x_n$ and $m$ clauses $C_1, \dots, C_m$ where each $C_i = \ell_i^1 \lor \ell_i^2 \lor \ell_i^3$
with $\ell_i^1, \ell_i^2, \ell_i^3 \in \{x_1,\dots,x_n\} \cup \{\overline{x}_1,\dots,\overline{x}_n\}$. Furthermore each literal $x_i, \overline{x}_i$ appears in at most 4 clauses. For each $1 \le i \le n$, let $\#(x_i), \#(\overline{x}_i)$ be the number of times the literals $x_i, \overline{x}_i$ appear in the formula $\varphi$, respectively. Additionally, to each literal \(\ell\in \{x_1,\dots,x_n\} \cup \{\overline{x}_1,\dots,\overline{x}_n\} \) we assign a distinct prime number $P(\ell) \ge 5$. By the prime number theorem the $n^{th}$ prime is of magnitude $O(n \log n)$. Hence all these primes can be written down in unary in polynomial time in the size of $n$, by running a sieve algorithm such as the Sieve of Eratosthenes.

We now encode an assignment to the variables $x_1,\dots,x_n$ by means of a \emph{sum of reciprocal of the prime numbers}. That is, we map an assignment $A$ to the number $f(A) := \sum_{i=1}^n \frac{f_i}{r_i}$ where $(f_i, r_i) = (\#(x_i),P(x_i))$ if $A(x_i)$ is true and $(f_i, r_i) = (\#(\overline{x}_i),P(\overline{x}_i))$ if $A(x_i)$ is false. We call this the \emph{Egyptian prime fraction} encoding, named after the Egyptian fractions technique of using sums of fractions of the form $1/n$ to encode non-negative rational numbers.
As the following proposition shows, this way of using sums of reciprocals of primes to encode assignments is unique. 

\begin{proposition}\label{prop:Egyptian-prime-encoding}
	Suppose $\sum_{i=1}^k \frac{f_i}{r_i} = \sum_{j=1}^\ell \frac{g_j}{s_j}$ where $\{r_1,\dots,r_k\}$ (resp. $\{s_1,\dots,s_\ell\}$) is a set of pairwise distinct prime numbers greater or equal to \(5\)
	and $f_i,g_j \le 4$ for all \( i\in [k], j\in [\ell]\). Then, $\{r_1,\dots,r_k\} = \{s_1,\dots,s_\ell\}$. 
\end{proposition}

\begin{proof}
	Suppose $\sum_{i=1}^k \frac{f_i}{r_i} = \sum_{j=1}^\ell \frac{g_j}{s_j}$ and suppose $r_d \notin \{s_1,\dots,s_\ell\}$ for some $d$.
	Hence, we get that $\frac{f_d}{r_d} = \sum_{j=1}^\ell \frac{g_j}{s_j} - \sum_{i=1, i\neq d}^k \frac{f_i}{r_i}$. Multiplying both sides  by $P = \prod_{i=1}^k r_i \cdot \prod_{j=1}^\ell s_j$ we get
	$\frac{Pf_d}{r_d} = \sum_{j=1}^\ell \frac{Pg_j}{s_j} - \sum_{i=1, i \neq d}^k \frac{Pf_i}{r_i}$. Notice that each summand on the RHS is a natural number divisible by $r_d$, but the LHS is not divisible by $r_d$. Hence $\{r_1,\dots,r_k\} \subseteq \{s_1,\dots,s_\ell\}$. By a symmetric proof, we also have $\{s_1,\dots,s_\ell\} \subseteq \{r_1,\dots,r_k\}$ and so they are the same.
\end{proof}

This proposition immediately implies that $f(A) = f(A')$ iff $A = A'$ and therefore, our encoding of the assignments is unique. With this encoding in mind, we now construct an acyclic unary 1-CVASS $\mach$  which consists of two gadgets, an \emph{assignment gadget} $\mach_A$ which encodes an assignment using the above encoding, and a \emph{clause gadget} $\mach_C$ which checks if this assignment satisfies all the clauses (see also Fig.~\ref{fig-mach-CVASS}).

%
%

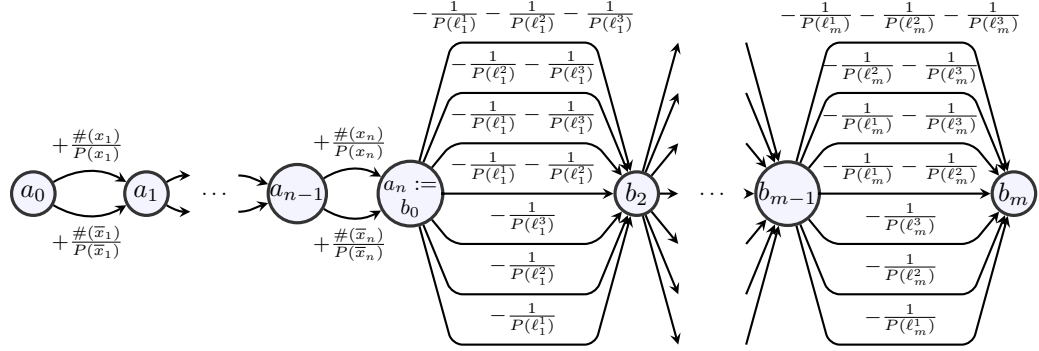
\begin{figure}[htbp]
	\centering
	\begin{tikzpicture}[scale=.95, every node/.style={scale=0.95}, x=2.1cm, y=2.1cm, font=\footnotesize]
		
		\pgfmathsetmacro{\nodedist}{0.75}
		\pgfmathsetmacro{\bnodedist}{1.5}
		\pgfmathsetmacro{\dotsdist}{0.5}
		
		\pgfmathsetmacro{\dotTempNodessDist}{0.35}
		
		\pgfmathsetmacro{\bdotsdist}{0.5}
		\pgfmathsetmacro{\xb}{0.75+\nodedist+2*\dotsdist}
		
		\node[state] (a0) at (0,0) {\(a_0\)};
		\node[state] (a1) at (\nodedist,0) {\(a_1\)};
		\node[] (dots) at (\nodedist+\dotsdist-0.05,0) {\(\dots\)};
		\node[state] (an) at (\nodedist+2*\dotsdist,0) {\(a_{n-1}\)};
		
		\node[state, minimum width=0.9cm, minimum height=0.9cm] (b1) at (\xb,0) {};
		\node[anchor=north, yshift = -4.8pt] at (b1.north) {\(a_n :=\)};
		\node[anchor=south, yshift = 0.7pt] at (b1.south) {\(b_0\)};
		

		\node[state] (b2) at (\xb+\bnodedist,0) {\(b_2\)};
		\node[] (bdots) at (\xb+\bnodedist+\bdotsdist,0) {\(\dots\)};
		\node[state] (bm) at (\xb+\bnodedist+2*\bdotsdist,0) {\(b_{m-1}\)};	
		\node[state] (bmplus) at (\xb+2*\bnodedist+2*\bdotsdist,0) {\(b_{m}\)};

		\draw [tran,bend left] (a0) to node[above] {\parbox{3.3em}{\centering \(+\frac{\#(x_1)}{P(x_1)}\)}} (a1);
		\draw [tran,bend right] (a0) to node[below] {\parbox{3.3em}{\centering \(+\frac{\#(\overline{x}_1)}{P(\overline{x}_1)}\)}} (a1);
		
		\draw [tran,bend left] (an) to node[above] {\parbox{3.3em}{\centering \(+\frac{\#(x_n)}{P(x_n)}\)}} (b1);
		\draw [tran,bend right] (an) to node[below] {\parbox{3.3em}{\centering \(+\frac{\#(\overline{x}_n)}{P(\overline{x}_n)}\)}} (b1);
		
		%

		\node[] (a1uptemp) at (\nodedist+\dotTempNodessDist,0.12) {};
		\node[] (a1downtemp) at (\nodedist+\dotTempNodessDist,-0.12) {};
		
		\node[] (anuptemp) at (\nodedist+2*\dotsdist-\dotTempNodessDist-0.1,0.12) {};
		\node[] (andowntemp) at (\nodedist+2*\dotsdist-\dotTempNodessDist-0.1,-0.12) {};
		
		\draw [tran,out=30, in=180] (a1) to  (a1uptemp);
		\draw [tran,out=-30, in=180] (a1) to  (a1downtemp);
		
		\draw [tran,out=0, in=162] (anuptemp) to  (an);
		\draw [tran,out=0, in=198] (andowntemp) to  (an);

		\pgfmathsetmacro{\xtikz}{0.275}
		\pgfmathsetmacro{\ytikz}{\xb+\bnodedist}
		\pgfmathsetmacro{\ymtikz}{\xb+2*\bnodedist+2*\bdotsdist}
		\pgfmathsetmacro{\heighttikz}{1}
		\pgfmathsetmacro{\steptikz}{\heighttikz/3}
		
		\draw[tran,rounded corners] (b1) -- (\xb+\xtikz,\heighttikz-\steptikz*6)   -- node[above] {\parbox{3.3em}{\centering \(-\frac{1}{P(\ell_1^1)}\)}} (\ytikz-\xtikz,\heighttikz-\steptikz*6) -- (b2);
		
		\draw[tran,rounded corners] (b1) -- (\xb+\xtikz,\heighttikz-\steptikz*5)   -- node[above] {\parbox{3.3em}{\centering \(-\frac{1}{P(\ell_1^2)}\)}} (\ytikz-\xtikz,\heighttikz-\steptikz*5) -- (b2);
		
		\draw[tran,rounded corners] (b1) -- (\xb+\xtikz,\heighttikz-\steptikz*4)   -- node[above] {\parbox{3.3em}{\centering \(-\frac{1}{P(\ell_1^3)}\)}} (\ytikz-\xtikz,\heighttikz-\steptikz*4) -- (b2);	
		
		\draw[tran,rounded corners] (b1) -- (\xb+\xtikz,\heighttikz-\steptikz*3)   -- node[above] {\parbox{20em}{\centering \(-\frac{1}{P(\ell_1^1)}-\frac{1}{P(\ell_1^2)}\)}} (\ytikz-\xtikz,\heighttikz-\steptikz*3) -- (b2);	
		
		\draw[tran,rounded corners] (b1) -- (\xb+\xtikz,\heighttikz-\steptikz*2)   -- node[above] {\parbox{20em}{\centering \(-\frac{1}{P(\ell_1^1)}-\frac{1}{P(\ell_1^3)}\)}} (\ytikz-\xtikz,\heighttikz-\steptikz*2) -- (b2);	
		
		\draw[tran,rounded corners] (b1) -- (\xb+\xtikz,\heighttikz-\steptikz*1)   -- node[above] {\parbox{20em}{\centering \(-\frac{1}{P(\ell_1^2)}-\frac{1}{P(\ell_1^3)}\)}} (\ytikz-\xtikz,\heighttikz-\steptikz*1) -- (b2);	
		
		\draw[tran,rounded corners] (b1) -- (\xb+\xtikz,\heighttikz-\steptikz*0)   -- node[above] {\parbox{20em}{\centering \(-\frac{1}{P(\ell_1^1)}-\frac{1}{P(\ell_1^2)}-\frac{1}{P(\ell_1^3)}\)}} (\ytikz-\xtikz,\heighttikz-\steptikz*0) -- (b2);

		\draw[tran,rounded corners] (bm) -- (\xb+\bnodedist+2*\bdotsdist+\xtikz,\heighttikz-\steptikz*6)   -- node[above] {\parbox{3.3em}{\centering \(-\frac{1}{P(\ell_m^1)}\)}} (\ymtikz-\xtikz,\heighttikz-\steptikz*6) -- (bmplus);
		
		\draw[tran,rounded corners] (bm) -- (\xb+\bnodedist+2*\bdotsdist+\xtikz,\heighttikz-\steptikz*5)   -- node[above] {\parbox{3.3em}{\centering \(-\frac{1}{P(\ell_m^2)}\)}} (\ymtikz-\xtikz,\heighttikz-\steptikz*5) -- (bmplus);
		
		\draw[tran,rounded corners] (bm) -- (\xb+\bnodedist+2*\bdotsdist+\xtikz,\heighttikz-\steptikz*4)   -- node[above] {\parbox{3.3em}{\centering \(-\frac{1}{P(\ell_m^3)}\)}} (\ymtikz-\xtikz,\heighttikz-\steptikz*4) -- (bmplus);	
		
		\draw[tran,rounded corners] (bm) -- (\xb+\bnodedist+2*\bdotsdist+\xtikz,\heighttikz-\steptikz*3)   -- node[above] {\parbox{20em}{\centering \(-\frac{1}{P(\ell_m^1)}-\frac{1}{P(\ell_m^2)}\)}} (\ymtikz-\xtikz,\heighttikz-\steptikz*3) -- (bmplus);

		\draw[tran,rounded corners] (bm) -- (\xb+\bnodedist+2*\bdotsdist+\xtikz,\heighttikz-\steptikz*2)   -- node[above] {\parbox{20em}{\centering \(-\frac{1}{P(\ell_m^1)}-\frac{1}{P(\ell_m^3)}\)}} (\ymtikz-\xtikz,\heighttikz-\steptikz*2) -- (bmplus);	
		
		\draw[tran,rounded corners] (bm) -- (\xb+\bnodedist+2*\bdotsdist+\xtikz,\heighttikz-\steptikz*1)   -- node[above] {\parbox{20em}{\centering \(-\frac{1}{P(\ell_m^2)}-\frac{1}{P(\ell_m^3)}\)}} (\ymtikz-\xtikz,\heighttikz-\steptikz*1) -- (bmplus);	
		
		\draw[tran,rounded corners] (bm) -- (\xb+\bnodedist+2*\bdotsdist+\xtikz,\heighttikz-\steptikz*0)   -- node[above] {\parbox{20em}{\centering \(-\frac{1}{P(\ell_m^1)}-\frac{1}{P(\ell_m^2)}-\frac{1}{P(\ell_m^3)}\)}} (\ymtikz-\xtikz,\heighttikz-\steptikz*0) -- (bmplus);

		\draw[tran,rounded corners] (b2)  -- (\ytikz+\xtikz,\heighttikz-\steptikz*0) ;	
		\draw[tran,rounded corners] (b2)  -- (\ytikz+\xtikz,\heighttikz-\steptikz*1) ;			
		\draw[tran,rounded corners] (b2)  -- (\ytikz+\xtikz,\heighttikz-\steptikz*2) ;			
		\draw[tran,rounded corners] (b2)  -- (\ytikz+\xtikz,\heighttikz-\steptikz*3) ;			
		\draw[tran,rounded corners] (b2)  -- (\ytikz+\xtikz,\heighttikz-\steptikz*4) ;			
		\draw[tran,rounded corners] (b2)  -- (\ytikz+\xtikz,\heighttikz-\steptikz*5) ;			
		\draw[tran,rounded corners] (b2)  -- (\ytikz+\xtikz,\heighttikz-\steptikz*6) ;

		\draw[tran,rounded corners]  (\xb+\bnodedist+2*\bdotsdist-\xtikz,\heighttikz)   -- (bm);		
		\draw[tran,rounded corners]  (\xb+\bnodedist+2*\bdotsdist-\xtikz,\heighttikz-\steptikz)  -- (bm);		
		\draw[tran,rounded corners]  (\xb+\bnodedist+2*\bdotsdist-\xtikz,\heighttikz-2*\steptikz)  -- (bm);		
		\draw[tran,rounded corners]  (\xb+\bnodedist+2*\bdotsdist-\xtikz,\heighttikz-3*\steptikz)  -- (bm);		
		\draw[tran,rounded corners]  (\xb+\bnodedist+2*\bdotsdist-\xtikz,\heighttikz-4*\steptikz)  -- (bm);		
		\draw[tran,rounded corners]  (\xb+\bnodedist+2*\bdotsdist-\xtikz,\heighttikz-5*\steptikz)  -- (bm);		
		\draw[tran,rounded corners]  (\xb+\bnodedist+2*\bdotsdist-\xtikz,\heighttikz-6*\steptikz)  -- (bm);		
	\end{tikzpicture}
	\caption{The CVASS $\mach$. The first half (from $a_0$ to $a_n$) depicts the assignment gadget $\mach_A$. The second half (from $a_n$ to $b_m$) depicts the clause gadget.}
	\label{fig-mach-CVASS}
\end{figure}

\subparagraph*{The assignment gadget.}

 $\mach_A$ has states \(a_0,\dots,a_n \), and for each $0 \le i < n$ there are two transitions \((a_i,\#(x_i)/P(x_i),a_{i+1}) \) and \((a_i,\#(\overline{x}_i)/P(\overline{x}_i),a_{i+1}) \). See the first half of Fig.~\ref{fig-mach-CVASS} for a representation of $\mach_A$. It is then immediately clear that there is a restricted $\qn$-run from $a_0(0)$ to $a_n(k)$ iff  
$k = f(A)$ for some assignment $A$. 



\subparagraph*{The clause gadget.}


 The clause gadget $\mach_C$ checks whether the assignment $A$ generated by the assignment gadget $\mach_A$ satisfies the formula $\varphi$. It does this, by guessing for each clause $C_i$, the subset of literals $\ell_i^1, \ell_i^2, \ell_i^3$ that are satisfied by $A$ and then decrementing the counter by the appropriate inverse of the primes.
For example, if it guesses that the assignment satisfies only $\ell_i^2, \ell_i^3$, it will decrement
the counter by $1/P(\ell_i^2) + 1/P(\ell_i^3)$. 

More precisely $\mach_C$ has the
states \( b_0, b_1,\dots,b_m\) and for each $0 \le i < m$ and for each of the 7 non-empty subsets \(S_i\subseteq\{\ell_i^1,\ell_i^2,\ell^3_i \} \), $\mach_C$ has a transition $(b_i,-\sum_{\ell\in S_i}1/P(\ell), b_{i+1})$. See the second half of Fig.~\ref{fig-mach-CVASS} for a representation of $\mach_C$.
We have the following proposition
which asserts the main property of the clause gadget that we want.


\begin{restatable}{proposition}{propositionrestrictedreachabilityclausegadget}
\label{proposition-restricted-reachability-clause-gadget}
	There is a restricted $\qn$-run from $b_0(f(A))$ to $b_m(0)$ iff $A$ satisfies the formula $\varphi$.
\end{restatable}
 The \(\Rightarrow\) direction holds from the uniqueness of encodings given by Proposition~\ref{prop:Egyptian-prime-encoding}. For the	\(\Leftarrow\) direction it suffices to always pick the transition decrementing the counter by $\sum_{\ell \in S_i} \frac{1}{P(\ell)}$ for the non-empty set of literals of \(C_i \) that are true per \(A\). A full proof of Proposition~\ref{proposition-restricted-reachability-clause-gadget} can be found in Appendix~\ref{app-proposition-restricted-reachability-clause-gadget-proof}.

Let $\mach$ be the acyclic 1-CVASS we obtain by combining $\mach_A$ and $\mach_C$ by fusing together the states $a_n$ and $b_0$ (See Figure~\ref{fig-mach-CVASS}). Notice that any $\qn$-run in $\mach$ from $a_0(0)$ to $b_m(0)$ must also be a $\qnz$-run: Indeed,
in any such run, the counter is only incremented in the stretch between $a_0$ to $a_n$
and only decremented in the stretch between $a_n$ and $b_m$.  By the above properties of $\mach_A$ and $\mach_C$ we hence get that

\begin{theorem}\label{thm-2cvass-restricted-NP-hard}
	There is a restricted $\qn$-run (resp. restricted $\qnz$-run) from $a_0(0)$ to $b_m(0)$ in $\mach$ iff $\varphi$ is satisfiable.
\end{theorem}




\subsection*{Stage 2: Reducing Restricted Runs to Coverability}

We now show how to get rid of the restricted run assumption in the 1-CVASS $\mach$ from the previous subsection.
First, by construction, notice that any run in $\mach$ from $a_0(0)$ to $b_m(0)$ takes exactly $K := n+m$
steps. Furthermore, any single transition of $\mach$ can update the counter by a number whose absolute value is at most $3$. Indeed, the largest possible absolute value appearing in the transitions is either of the form $\frac{1}{p} + \frac{1}{q} + \frac{1}{r}$  or of the form $\frac{d}{p}$ for some $d \le 4$. 
Hence, it follows that the maximum absolute value attainable in the counter in any run is at most $3K$.

We exploit this observation and modify $\mach$ into an acyclic 2-CVASS $\mach'$ by adding one extra counter to  $\mach$, and changing each transition $(p,w,q)$ of \(\mach \) into $(p,(w,1-w),q)$  in \(\mach' \).
We can now prove that

\begin{restatable}{lemma}{lemmacoverabilitycvassnpcomplete}
\label{lemma-coverability-2cvass-npcomplete}
	There is a restricted $\qn$-run (resp. $\qnz$-run) from $a_0(0)$ to $b_m(0)$ in $\mach$ iff there is a $\qn$-run (resp. $\qnz$-run) from $a_0(0,3K)$ covering $b_m(0,4K)$ in $\mach'$.
\end{restatable}

Intuitively, the reason behind this lemma is that each transition of $\mach'$ increases the sum of the two counters by $1$. Hence, by requiring that the sum of the two counters at the very end has increased by at least $K$, we force
the run in $\mach'$ satisfying that requirement to fire each step with the fraction \(1\), thereby inducing a restricted run in $\mach$. 
The full proof can be found in Appendix~\ref{app-lemma-coverability-2cvass-npcomplete-proof}.

Combined with Theorem~\ref{thm-2cvass-restricted-NP-hard}, this then allows us to conclude that 2-$\Cover_\qn$ and 2-$\Cover_{\qnz}$ are both $\NP$-hard over unary encoding. This proves the second part of Theorem~\ref{thm:2-CVASS}, which completes the main result of this section.
\section{Lower Bounds for 3-CVASS (with Integer Updates)}
\label{sec-3CVASS}
In this section, we prove Theorem~\ref{thm:main-result-integer-updates-3-CVASS}. To this end, we prove the following main result.
\begin{theorem}\label{thm:3-CVASS}
    3-$\Cover_{\qnz}$ for acyclic unary 3-CVASS with integer updates is $\NP$-hard.
\end{theorem}

By Remark~\ref{rem:cov-to-reach}, it is easily seen that this implies all of the $\NP$ lower bounds mentioned in Theorem~\ref{thm:main-result-integer-updates-3-CVASS}. 
As before, all of the $\NP$ upper bounds in Theorem~\ref{thm:main-result-integer-updates-3-CVASS} follow from~\cite[Section IV.B]{BlondinH17}. So the rest of this section is dedicated to proving Theorem~\ref{thm:3-CVASS}, which we do in two stages: First, we show that
coverability for acyclic unary 2-CVASS \emph{with zero-tests} over the $\qnz$-semantics
is $\NP$-hard. In the next stage, we show how to reduce this problem to coverability for acyclic unary 3-CVASS. 

Throughout this section, we shall exclusively consider the $\qnz$-semantics.

\subsection*{Stage 1: Hardness for Acyclic Unary 2-CVASS with Zero-Tests}

In this stage, we show $\NP$-hardness of coverability for an extended model of CVASS, namely CVASS with zero-tests. A $d$-CVASS with zero-tests is a triple $\mach = (Q,\transitions,Z)$ where $(Q,\transitions)$ is a $d$-CVASS and $Z : Q \to 2^{[d]}$ is a mapping which mandates that in a configuration $p(\bx)$, every counter in $Z(p)$ must be zero. Formally, a configuration of $\mach$ is a pair $p(\bx)$ where $p \in Q, \bx \in \qnz^d$ and $\bx[i] = 0$ for every $i \in Z(p)$. 
Given two configurations $p(\bx)$ and $q(\by)$, a transition
$(p,\bw,q)$ and a non-zero fraction $\alpha \in (0,1]$, we say that $p(\bx) \act{\alpha \bw} q(\by)$
if $\bx[i] = 0$ for all $i \in Z(p)$, $\by = \bx + \alpha \bw$ and $\by[i] = 0$ for all $i \in Z(q)$.
Similarly to CVASS, we define the notions of steps, $\qnz$-runs, and $\qnz$-reachability for CVASS with zero-tests. Our main result in this stage is to show that

\begin{theorem}
	Reachability for acyclic unary 2-CVASS with integer updates and zero-tests over the $\qnz$-semantics is $\NP$-hard.
\end{theorem}

Similar to the previous section, we will prove this theorem by giving a reduction from a constrained version of 3-SAT, in which each variable appears at most 4 times in a formula. To this end, let $\varphi$ be a 3-CNF formula over $n$ variables $x_1,\dots,x_n$ and $m$ clauses $C_1, \dots, C_m$ where each $C_i = \ell_i^1 \lor \ell_i^2 \lor \ell_i^3$ with $\ell_i^1, \ell_i^2, \ell_i^3 \in \{x_1,\dots,x_n\} \cup \{\overline{x}_1,\dots,\overline{x}_n\}$. Furthermore, each literal $x_i, \overline{x}_i$ appears at most 4 times in $\varphi$. For each $1 \le i \le n$, let $\#(x_i), \#(\overline{x}_i)$ be the number of times the literals $x_i, \overline{x}_i$ appear in the formula $\varphi$, respectively. Additionally, to each literal \(\ell\in \{x_1,\dots,x_n\} \cup \{\overline{x}_1,\dots,\overline{x}_n\} \) we assign a distinct prime number $P(\ell) \ge 5$, such that $P(\ell)=P(\ell')$.
We can now uniquely encode an assignment \(A \) to the variables as the fraction $f(A) := \frac{1}{r_1^{f_1} \cdot r_2^{f_2} \cdot \dots \cdot r_n^{f_n}}$ where $(r_i, f_i) = (P(x_i),\#(x_i))$ if $A(x_i)$ is true and $(r_i,f_i) = (P(\overline{x}_i),\#(\overline{x}_i))$ if $A(x_i)$ is false.

\begin{figure}[htb]
	\centering
	\begin{tikzpicture}[scale=.95, every node/.style={scale=0.95}, x=2.1cm, y=2.1cm, font=\footnotesize]
		\begin{scope}[shift={(0,0)}]
			\node[state] (a1) at (0,0) {\shortstack{\(a_{i}\)\\ \(y=0\)}};
			\node[state] (a2) at (2.25,0) {\shortstack{\(a_{i+1}\)\\ \(y=0\)}};

						\node[state, minimum width=1.075cm, minimum height=1.075cm] (b1) at (1.25,0) {};
						\node[anchor=north, yshift = 0pt,font=\large] at (b1.north) {\(a_i'\)};
						\node[anchor=south, yshift = 5pt,font=\large] at (b1.south) {\(x=0\)};

			\draw [tran,bend left] (a1) to node[above, align=center, text width=8em] { \(\big(-P(x_i)^{\#(x_i)},1\big)\)} (b1);
			
			\draw [tran,bend right] (a1) to node[below, align=center, text width=8em] { \(\big(-P(\overline{x}_i)^{\#(\overline{x}_i)},1\big)\)} (b1);
			
			\draw [tran] (b1) to node[above] {\parbox{3.3em}{\centering \((1,-1)\)}} (a2);
			
		\end{scope}
		\begin{scope}[shift={(3.25,0)}]
			
			\pgfmathsetmacro{\xtikz}{0.34}
			\pgfmathsetmacro{\ytikz}{1.5}
			
			\pgfmathsetmacro{\heighttikz}{1}
			\pgfmathsetmacro{\steptikz}{\heighttikz/3}
			
			\pgfmathsetmacro{\xfirstnode}{0}
			\pgfmathsetmacro{\xmidnodes}{1}
			\pgfmathsetmacro{\xlastnode}{3.2}

			\node[state] (s1) at (\xfirstnode,0) {\shortstack{\(b_{i}\)\\ \(y=0\)}};
			\node[state] (s2) at (\xlastnode,0) {\shortstack{\(b_{i+1}\)\\ \(y=0\)}};
			

\node[state, minimum width=1.1cm, minimum height=1.1cm] (u) at (\xmidnodes,0) {};
\node[anchor=north, yshift = -1.45pt,font=\large] at (u.north) {\(b_i'\)};
\node[anchor=south, yshift = 5pt,font=\large] at (u.south) {\(x=0\)};
			
			\draw[tran] (s1) -- node[above] {\parbox{3.3em}{\centering \((-1,1)\)}} (u);

			\foreach \i/\ival in {
				0/{$(P(\ell_i^1) \cdot P(\ell_i^2)\cdot P(\ell_i^3),-1)$},
				1/{$\big(P(\ell_i^1) \cdot P(\ell_i^2),-1\big)$},
				2/{$\big(P(\ell_i^1) \cdot P(\ell_i^3),-1\big)$},
				3/{$\big(P(\ell_i^2) \cdot P(\ell_i^3),-1\big)$},
				4/{$\big(P(\ell_i^1),-1\big)$},
				5/{$\big(P(\ell_i^2),-1\big)$},
				6/{$\big(P(\ell_i^3),-1\big)$}
			} {
				\draw[tran,rounded corners] (u) -- (\xmidnodes+\xtikz,\heighttikz-\steptikz*\i) -- node[above] {\parbox{12em}{\centering \ival}} (\xlastnode-\xtikz,\heighttikz-\steptikz*\i) -- (s2);
			}
			
		\end{scope}
	\end{tikzpicture}
	\caption{A section of the assignment gadget (left) and the clause gadget (right). Here we use \(x=0\) (resp. \(y=0\)) inside a state to denote that that state zero-tests the first (resp. second) counter.}
	\label{fig-2CVASS-zero-tests-gadgets}
\end{figure}
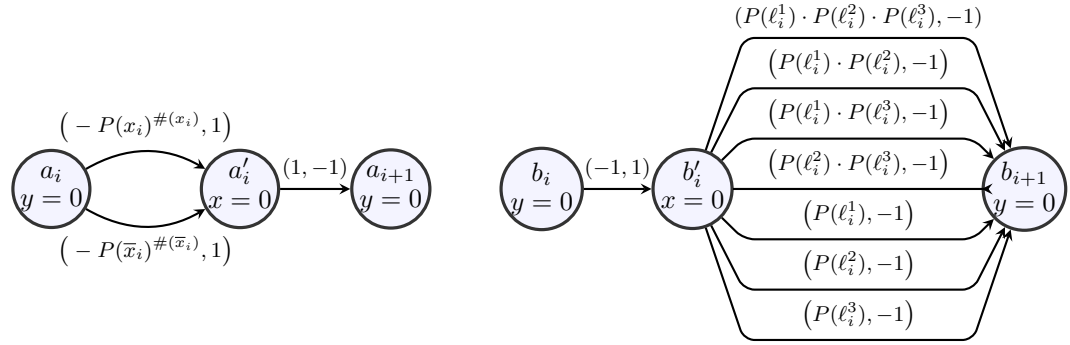

We now construct an acyclic unary 2-CVASS with integer updates and zero-tests $\mach$ which consists of two gadgets, the \emph{assignment gadget} $\mach_A$ which encodes an assignment, and a \emph{clause gadget} $\mach_C$ which checks if this assignment satisfies all the  clauses

\subparagraph*{The assignment gadget.}
 $\mach_A$ has states \(a_0,\dots,a_n,a_0',\dots,a_{n-1}' \) with corresponding zero-tests given by 
 $Z(a_0)=\dots=Z(a_n) = \{2\}$ and $Z(a_0')=\dots=Z(a_{n-1}') = \{1\}$. Furthermore, for each $0 \le i < n$, there are three transitions $(a_i, (-P(x_i)^{\#(x_i)},1), a_i'), (a_i, (-P(\overline{x}_i)^{\#(\overline{x}_i)},1), a_i')$
and $(a_i', (1,-1), a_{i+1})$. (See Fig.~\ref{fig-2CVASS-zero-tests-gadgets}).

Intuitively, the first two transitions divide the current 
value of the first counter by either $P(x_i)^{\#(x_i)}$ or  $P(\overline{x}_i)^{\#(\overline{x}_i)}$ while transferring this divided value to the second counter. The third transition then transfers this divided value back to the first counter. This effect of these transitions is formalized by the following proposition, whose proof follows by a case analysis of the possible fractions that can be used, under the influence of the zero-tests on the states.

\begin{restatable}{proposition}{assignmentgadget}\label{prop:assignment-gadget}
	For any $0 \le i < n$, there is a $\qnz$-run from $a_i(g_i,h_i)$ to $a_{i+1}(g_{i+1},h_{i+1})$ iff $h_i = h_{i+1} = 0$ and either $g_{i+1} = g_i/P(x_i)^{\#(x_i)}$ or $g_{i+1} = g_i/P(\overline{x}_i)^{\#(\overline{x}_i)}$.
\end{restatable}

The proof of Proposition~\ref{prop:assignment-gadget} can be found in Appendix~\ref{app-prop:assignment-gadget-proof}. Repeated application of this proposition implies that there is a $\qnz$-run from $a_0(1,0)$
to $a_n(g,h)$ for any $g,h$ iff $g = f(A)$ for some assignment $A$ and $h = 0$. Further,
since each $P(\ell) \in O(n \log n), \#(\ell) \le 4$, it follows that all of the numbers in the transitions can be encoded in unary in polynomial time.

%


\subparagraph*{The clause gadget.}
 The clause gadget $\mach_C$ checks whether the assignment $A$ generated by the assignment gadget $\mach_A$ satisfies the formula $\varphi$.  It does this, by guessing for each clause $C_i$, the subset of literals $\ell_i^1, \ell_i^2, \ell_i^3$ that are satisfied by $A$ and then executing a gadget that \emph{multiplies} the counter by the corresponding primes. More precisely $\mach_C$ has the states \( b_0,\dots,b_m, b_0',\dots,b_{m-1}'\) with corresponding zero-tests given by \(Z(b_0)=\dots=Z(b_m)=\{2 \}\)  and \(Z(b_0')=\dots=Z(b_{m-1}')=\{1\} \).  Furthermore, for each $0 \le i < m$, there is one transition  $(b_i,(-1,1),b_i')$ and seven other transitions given by $(b_i',(\prod_{\ell \in S_i} P(\ell), -1), b_{i+1})$ for
every non-empty subset $S_i\subseteq \{\ell_i^1,\ell_i^2,\ell_i^3\}$ (See Fig.~\ref{fig-2CVASS-zero-tests-gadgets}). 
We now have the following proposition, whose proof is similar to Proposition~\ref{prop:assignment-gadget}.

\begin{restatable}{proposition}{clausegadget}\label{prop:clause-gadget}
    For any $0 \le i < m$, there is a $\qnz$-run from $b_i(g_i,h_i)$ to $b_{i+1}(g_{i+1},h_{i+1})$ iff
    $h_i = h_{i+1} = 0, g_i \le 1$ 
    and $g_{i+1} = g_i \cdot p$ where $p$ is the product of some non-empty
    subset of $\{P(\ell_i^1),P(\ell_i^2),P(\ell_i^3)\}$. 
\end{restatable}
The proof of Proposition~\ref{prop:clause-gadget} can be found in Appendix~\ref{app-prop:clause-gadget-proof}. Hence, repeated applications of this proposition means that any run from $b_0$ to $b_m$ essentially multiplies
the first counter by a subset of the primes $\{P(\ell_i^1),P(\ell_i^2),P(\ell_i^3)\}$ for each $i$. This then implies that,

%

\begin{restatable}{lemma}{clauseapplication}\label{lemma-clauseapplication}
    There is a $\qnz$-run from $b_0(f(A),0)$ to $b_m(1,0)$ iff $A$ satisfies $\varphi$.
\end{restatable}

The proof of Lemma~\ref{lemma-clauseapplication} can be found in Appendix~\ref{app-lemma-clauseapplication-proof}. However, since our goal is to show hardness for coverability, we have to modify the condition in the above lemma so that it becomes a coverability condition. To that end, we add a state \(b_m' \) with zero test $Z(b_m') = \{1\}$ and a transition $(b_m,(-1,1),b_m')$ to $\mach_C$. It immediately follows that

\begin{restatable}{lemma}{coverability}\label{lemma-coverability}
    There is a $\qnz$-run from $b_0(f(A),0)$ covering $b_m'(0,1)$ iff $A$ satisfies $\varphi$. 
\end{restatable}
The proof of Lemma~\ref{lemma-coverability} can be found in the Appendix~\ref{app-lemma-coverability-proof}. It is now clear that if we let $\mach$ be the acyclic 2-CVASS with zero-tests obtained by combining $\mach_A$ and $\mach_C$ by fusing together $a_n$ and $b_0$, we get

%

\begin{theorem}\label{thm:satisfiable-with-zero-tests}
    There is a $\qnz$-run from $a_0(1,0)$ covering $b_m'(0,1)$ iff $\varphi$ is satisfiable.
\end{theorem}

This completes Stage 1 of the reduction. We now move on to Stage 2.

\subsection*{Stage 2: Removing Zero-Tests by One Extra Counter}

We  now show how to get rid of the zero-tests in $\mach$ by adding one extra counter. To this end,
we adapt the idea of the \emph{controlling counter} technique  from VASS~\cite{CzerwinskiO21}. The idea is as follows: Consider the CVASS $\overline{\mach}$ that we get from $\mach$ by simply removing its zero-tests.
As with $\mach$, any $\qnz$-run in $\overline{\mach}$ from $a_0(1,0)$ covering $b_m'(0,1)$ must be of the form 
\begin{multline}\label{eq:run-without-zero-tests}
    a_0(1,0) := a_0(g_0,h_0) \actqnz{\alpha_0 v_0} a_0'(g_0',h_0') \actqnz{\alpha_0' v_0'} a_1(g_1,h_1)  \cdots \actqnz{\alpha_{n-1}' v_{n-1}'}  a_n(g_n,h_n) := b_0(k_0,r_0)  \\  \actqnz{\beta_0 w_0} b_0'(k_0',r_0') \actqnz{\beta_0' w_0'} b_1(k_1,r_1) \dots \actqnz{\beta_m w_m} b_m(k_m,r_m) \actqnz{\beta_m' w_m'} b_m' (k_m',r_m') 
\end{multline}

with $k_m' \ge 0, r_m' \ge 1$.  This $\qnz$-run in $\overline{\mach}$ is also a $\qnz$-run in $\mach$ iff each $h_i, r_i = 0$ and each $g_i', k_i' = 0$. Let us call such $\qnz$-runs as \emph{good $\qnz$-runs} of $\overline{\mach}$. Hence a good $\qnz$-run of $\overline{\mach}$ exists iff there is a $\qnz$-run
of $\mach$ from $a_0(1,0)$ covering $b_m'(0,1)$.
Since we consider the \(\qnz \)-semantics it holds that $h_i, r_i, g_i', k_i' \geq  0$ for all \(i\). Hence all of them are 0 iff their combined sum is 0, i.e., a $\qnz$-run of the form Eq.~\ref{eq:run-without-zero-tests} is a good run iff 
\begin{equation}\label{eq:control}
 \sum_{i=0}^{n-1} g_i'  + \sum_{i=1}^n h_i + \sum_{i=0}^m k_i' + \sum_{i=1}^m r_i = 0   
\end{equation}

Now, we claim that the above equation can be described differently by expressing the LHS of Equation \ref{eq:control} in a different way. For instance, we know that each $g_i'$ is essentially the value of the first counter in the sum $(1,0) + \sum_{j = 0}^i \alpha_i v_i + \sum_{j=0}^{i-1} \alpha_i' v_i'$. Similarly, $h_i, k_i', r_i$ can also be written as such sums. Combining all of these sums into one, we can rewrite the above identity as
the one given by the following proposition.
\begin{restatable}{proposition}{cc}\label{prop:cc}
    A $\qnz$-run of the form Eq.~\ref{eq:run-without-zero-tests} is a good $\qnz$-run iff
\begin{multline}\label{eq-tester}
    n + (m+1) + \left(\sum_{i=0}^{n-1} \alpha_i \big((n-i+m) (v_i[1] + v_i[2]) + v_i[1]\big) \right) + 
    \left(\sum_{i=0}^{n-1} \alpha_i' (n-i+m) (v_i'[1] + v_i'[2]) \right) + \\  
    \left(\sum_{i=0}^{m} \beta_i \big((m-i) (w_i[1] + w_i[2]) + w_i[1]\big) \right) +
    \left(\sum_{i=0}^{m-1} \beta_i' (m-i) (w_i'[1] + w_i'[2]) \right) = 0
\end{multline}
\end{restatable}

The proof of Proposition~\ref{prop:cc} can be found in Appendix~\ref{app-prop:cc-proof}.
We call the number on the LHS of Eq.~\ref{eq-tester} as the \emph{tester} of the run. 
(We stress that this tester is the same as the value $\sum_{i=0}^{n-1} g_i'  + \sum_{i=1}^n h_i + \sum_{i=0}^m k_i' + \sum_{i=1}^m r_i$ and is hence always non-negative). By \Cref{prop:cc} it follows that
\begin{lemma}\label{lem:satisfiable-with-tester}
    There is a $\qnz$-run from $a_0(1,0)$ covering $b_m'(1,0)$ in $\mach$ iff there is a $\qnz$-run from
    $a_0(1,0)$ covering $b_m'(1,0)$ in $\overline{\mach}$ whose tester equals 0.
\end{lemma}

In order to remove the tester condition, we now introduce another counter to
keep track of the current value of the tester encountered so far. More precisely, for any prefix $P$ of a run of the form Eq.~\ref{eq:run-without-zero-tests}, we assign a number $C(P)$ as follows: For the empty prefix, $C(P)$ is simply $n+m+1$.
Now, for any prefix $P = P' s$, $C(P) = C(P') + C(s)$ where
\begin{itemize}
    \item $C(s) = \alpha_i \big((n-i+m) (v_i[1] + v_i[2]) + v_i[1]\big)$ if $s \equiv a_i(g_i,h_i) \act{\alpha_i v_i} a_i'(g_i',h_i')$.
    \item $C(s) = \alpha_i' (n-i+m) (v_i'[1] + v_i'[2])$ if $s \equiv a_i'(g_i',h_i') \act{\alpha_i' v_i'} a_{i+1}(g_{i+1},h_{i+1})$.
    \item $C(s) = \beta_i \big((m-i) (w_i[1] + w_i[2]) + w_i[1]\big)$ if $s \equiv b_i(k_i,r_i) \act{\beta_i w_i} b_i'(k_i',r_i')$.
    \item $C(s) = \beta_i' (m-i) (w_i'[1] + w_i'[2])$ if $s \equiv b_i'(k_i',r_i') \act{\beta_i' w_i'} b_{i+1}(k_{i+1},r_{i+1})$.
\end{itemize}

By construction, it is easily seen that $C(R)$ for a $\qnz$-run $R$ is precisely its tester. Hence, in order to compute $C(R)$, we add a third counter to $\overline{\mach}$ with initial value $n+m+1$, and at each step $s$, we simply add $C(s)$ to the third counter. We then finally require this third counter to have value 0. Whilst correct, this only gives a reachability condition on the third counter. Furthermore, while adding $C(s)$ at each step, the third counter might go below 0.

We now show how to avoid both these problems in one shot. To this end, let $K$ be the maximum \emph{absolute value} of all of the numbers in the transitions of $\overline{\mach}$ and let $L$ be 
$n+(m+1) + (n \cdot (n+m) \cdot 3K) + (n \cdot (n+m) \cdot 2K) + ((m+1) \cdot m \cdot 3K) + (m \cdot m \cdot 2K)$. Note that $L$ can be written down in unary in polynomial time.
By construction of $L$, it immediately follows that $C(P)$ for any prefix  has to lie in the range $[-L,L]$. Hence, if we now begin the third counter with the initial value $L - (n+m+1)$ and \emph{subtract} $C(s)$ at each step $s$, we are guaranteed that it will never go below zero. Furthermore, by construction, at the end of a run $R$,
the value of the third counter will be $L - C(R)$. As mentioned before, $C(R)$ is always non-negative and so the only way the third counter can reach a value $\ge L$ at the end is if $C(R) = 0$. This gives a coverability condition on the third counter that is equivalent to checking if the tester of the run is zero. If we call this new machine (with three counters) as $\mach'$ it then follows that

\begin{restatable}{lemma}{controllingcounter}\label{lem:satisfiable-with-controlling-counter}
    There is a $\qnz$-run from $a_0(1,0)$ covering $b_m'(0,1)$ in $\overline{\mach}$ whose tester equals 0 iff there is a $\qnz$-run from $a_0\big(1,0,L-(n+m+1)\big)$ covering $b_m'(0,1,L)$ in $\mach'$.
\end{restatable}

The proof of Lemma~\ref{lem:satisfiable-with-controlling-counter} (along with the formal construction of $\mach'$) can be found in Appendix~\ref{app-lem:satisfiable-with-controlling-counter-proof}.
By combining Theorem~\ref{thm:satisfiable-with-zero-tests} and Lemmas~\ref{lem:satisfiable-with-tester}~and~\ref{lem:satisfiable-with-controlling-counter} we get the following theorem.
\begin{theorem}
    There is a $\qnz$-run from $a_0\big(1,0,L-(n+m+1)\big)$ covering $b_m'(1,0,L)$ in $\mach'$ iff
    $\varphi$ is satisfiable.
\end{theorem}

This completes the proof of Theorem~\ref{thm:3-CVASS}, which is the main result of this section.

\section{Conclusion}

In this paper, we studied the complexity of reachability and coverability for fixed-dimensional CVASS over both unary/binary encoding as well as over both $\qn$-/$\qnz$- semantics. As our main result, we have completely characterized the complexity of all of these problems for all dimensions. Furthermore, we have also given an almost-complete classification in the case when the update vectors are restricted to be over the integers. One possible direction for future work is to study these problems for unary 2-CVASS with integer updates over $\qnz$ semantics as well as for unary $d$-CVASS with integer updates over $\qn$-semantics for any $d \ge 2$.



\bibliography{refs}

@article{siglog/Schmitz16,
  author       = {Sylvain Schmitz},
  title        = {The complexity of reachability in vector addition systems},
  journal      = {{ACM} {SIGLOG} News},
  volume       = {3},
  number       = {1},
  pages        = {4--21},
  year         = {2016},
  url          = {https://doi.org/10.1145/2893582.2893585},
  doi          = {10.1145/2893582.2893585},
  bibsource    = {dblp computer science bibliography, https://dblp.org}
}

@article{EsparzaGLM17,
  author       = {Javier Esparza and
                  Pierre Ganty and
                  J{\'{e}}r{\^{o}}me Leroux and
                  Rupak Majumdar},
  title        = {Verification of population protocols},
  journal      = {Acta Informatica},
  volume       = {54},
  number       = {2},
  pages        = {191--215},
  year         = {2017},
  url          = {https://doi.org/10.1007/s00236-016-0272-3},
  doi          = {10.1007/S00236-016-0272-3},
  bibsource    = {dblp computer science bibliography, https://dblp.org}
}

@article{GermanS92,
  author       = {Steven M. German and
                  A. Prasad Sistla},
  title        = {Reasoning about Systems with Many Processes},
  journal      = {J. {ACM}},
  volume       = {39},
  number       = {3},
  pages        = {675--735},
  year         = {1992},
  url          = {https://doi.org/10.1145/146637.146681},
  doi          = {10.1145/146637.146681},
  bibsource    = {dblp computer science bibliography, https://dblp.org}
}

@article{Aalst98,
  author       = {Wil M. P. van der Aalst},
  title        = {The Application of {P}etri Nets to Workflow Management},
  journal      = {J. Circuits Syst. Comput.},
  volume       = {8},
  number       = {1},
  pages        = {21--66},
  year         = {1998},
  url          = {https://doi.org/10.1142/S0218126698000043},
  doi          = {10.1142/S0218126698000043},
  bibsource    = {dblp computer science bibliography, https://dblp.org}
}

@inproceedings{Leroux21,
  author       = {J{\'{e}}r{\^{o}}me Leroux},
  title        = {The Reachability Problem for {P}etri Nets is Not Primitive Recursive},
  booktitle    = {62nd {IEEE} Annual Symposium on Foundations of Computer Science, {FOCS}
                  2021, Denver, CO, USA, February 7-10, 2022},
  pages        = {1241--1252},
  publisher    = {{IEEE}},
  year         = {2021},
  url          = {https://doi.org/10.1109/FOCS52979.2021.00121},
  doi          = {10.1109/FOCS52979.2021.00121},
  bibsource    = {dblp computer science bibliography, https://dblp.org}
}

@inproceedings{CzerwinskiO21,
  author       = {Wojciech Czerwinski and
                  Lukasz Orlikowski},
  title        = {Reachability in Vector Addition Systems is Ackermann-complete},
  booktitle    = {62nd {IEEE} Annual Symposium on Foundations of Computer Science, {FOCS}
                  2021, Denver, CO, USA, February 7-10, 2022},
  pages        = {1229--1240},
  publisher    = {{IEEE}},
  year         = {2021},
  url          = {https://doi.org/10.1109/FOCS52979.2021.00120},
  doi          = {10.1109/FOCS52979.2021.00120},
  bibsource    = {dblp computer science bibliography, https://dblp.org}
}

@article{Blondin20,
  author       = {Michael Blondin},
  title        = {The {ABC}s of {P}etri net reachability relaxations},
  journal      = {{ACM} {SIGLOG} News},
  volume       = {7},
  number       = {3},
  pages        = {29--43},
  year         = {2020},
  url          = {https://doi.org/10.1145/3436980.3436984},
  doi          = {10.1145/3436980.3436984},
  bibsource    = {dblp computer science bibliography, https://dblp.org}
}

@article{BlondinFHH17,
  author       = {Michael Blondin and
                  Alain Finkel and
                  Christoph Haase and
                  Serge Haddad},
  title        = {The Logical View on Continuous {P}etri Nets},
  journal      = {{ACM} Trans. Comput. Log.},
  volume       = {18},
  number       = {3},
  pages        = {24:1--24:28},
  year         = {2017},
  url          = {https://doi.org/10.1145/3105908},
  doi          = {10.1145/3105908},
  bibsource    = {dblp computer science bibliography, https://dblp.org}
}

@inproceedings{EsparzaLMMN14,
  author       = {Javier Esparza and
                  Rusl{\'{a}}n Ledesma{-}Garza and
                  Rupak Majumdar and
                  Philipp J. Meyer and
                  Filip Niksic},
  editor       = {Armin Biere and
                  Roderick Bloem},
  title        = {An {SMT}-Based Approach to Coverability Analysis},
  booktitle    = {Computer Aided Verification - 26th International Conference, {CAV}
                  2014, Held as Part of the Vienna Summer of Logic, {VSL} 2014, Vienna,
                  Austria, July 18-22, 2014. Proceedings},
  series       = {Lecture Notes in Computer Science},
  pages        = {603--619},
  publisher    = {Springer},
  year         = {2014},
  url          = {https://doi.org/10.1007/978-3-319-08867-9\_40},
  doi          = {10.1007/978-3-319-08867-9\_40},
  bibsource    = {dblp computer science bibliography, https://dblp.org}
}

@book{david2005discrete,
  title={Discrete, continuous, and hybrid {P}etri nets},
  author={David, Ren{\'e} and Alla, Hassane},
  year={2005},
  publisher={Springer},
  doi={10.1007/978-3-642-10669-9}
}

@inproceedings{BlondinH17,
  author       = {Michael Blondin and
                  Christoph Haase},
  title        = {Logics for continuous reachability in {P}etri nets and vector addition
                  systems with states},
  booktitle    = {32nd Annual {ACM/IEEE} Symposium on Logic in Computer Science, {LICS}
                  2017, Reykjavik, Iceland, June 20-23, 2017},
  pages        = {1--12},
  publisher    = {{IEEE} Computer Society},
  year         = {2017},
  url          = {https://doi.org/10.1109/LICS.2017.8005068},
  doi          = {10.1109/LICS.2017.8005068},
  bibsource    = {dblp computer science bibliography, https://dblp.org}
}

@INPROCEEDINGS{client-server,
  author={Mahulea, C. and Recalde, L. and Silva, M.},
  booktitle={2006 8th International Workshop on Discrete Event Systems}, 
  title={On performance monotonicity and basic servers semantics of continuous {P}etri nets}, 
  year={2006},
  volume={},
  number={},
  pages={345-351},
  doi={10.1109/WODES.2006.382530}}

@inproceedings{HaarK25,
  author       = {Stefan Haar and
                  Juraj Kolc{\'{a}}k},
  editor       = {Fran{\c{c}}ois Fages and
                  Sabine P{\'{e}}r{\`{e}}s},
  title        = {Continuous Petri Nets Faithfully Fluidify Most Permissive Boolean
                  Networks},
  booktitle    = {Computational Methods in Systems Biology - 23rd International Conference,
                  {CMSB} 2025, Lyon, France, September 10-12, 2025, Proceedings},
  series       = {Lecture Notes in Computer Science},
  pages        = {89--105},
  publisher    = {Springer},
  year         = {2025},
  url          = {https://doi.org/10.1007/978-3-032-01436-8\_6},
  doi          = {10.1007/978-3-032-01436-8\_6},
  bibsource    = {dblp computer science bibliography, https://dblp.org}
}

@inproceedings{HeinerGD08,
  author       = {Monika Heiner and
                  David R. Gilbert and
                  Robin Donaldson},
  editor       = {Marco Bernardo and
                  Pierpaolo Degano and
                  Gianluigi Zavattaro},
  title        = {Petri Nets for Systems and Synthetic Biology},
  booktitle    = {Formal Methods for Computational Systems Biology, 8th International
                  School on Formal Methods for the Design of Computer, Communication,
                  and Software Systems, {SFM} 2008, Bertinoro, Italy, June 2-7, 2008,
                  Advanced Lectures},
  series       = {Lecture Notes in Computer Science},
  pages        = {215--264},
  publisher    = {Springer},
  year         = {2008},
  url          = {https://doi.org/10.1007/978-3-540-68894-5\_7},
  doi          = {10.1007/978-3-540-68894-5\_7},
  bibsource    = {dblp computer science bibliography, https://dblp.org}
}

@article{HerajyH18,
  author       = {Mostafa Herajy and
                  Monika Heiner},
  title        = {Adaptive and Bio-semantics of Continuous Petri Nets: Choosing the
                  Appropriate Interpretation},
  journal      = {Fundam. Informaticae},
  volume       = {160},
  number       = {1-2},
  pages        = {53--80},
  year         = {2018},
  url          = {https://doi.org/10.3233/FI-2018-1674},
  doi          = {10.3233/FI-2018-1674},
  bibsource    = {dblp computer science bibliography, https://dblp.org}
}

@article{BalasubramanianER23,
  author       = {A. R. Balasubramanian and
                  Javier Esparza and
                  Mikhail A. Raskin},
  title        = {Finding Cut-Offs in Leaderless Rendez-Vous Protocols is Easy},
  journal      = {Log. Methods Comput. Sci.},
  volume       = {19},
  number       = {4},
  year         = {2023},
  url          = {https://doi.org/10.46298/lmcs-19(4:2)2023},
  doi          = {10.46298/LMCS-19(4:2)2023},
  bibsource    = {dblp computer science bibliography, https://dblp.org}
}

@inproceedings{Balasubramanian20,
  author       = {A. R. Balasubramanian},
  editor       = {Nitin Saxena and
                  Sunil Simon},
  title        = {Parameterized Complexity of Safety of Threshold Automata},
  booktitle    = {40th {IARCS} Annual Conference on Foundations of Software Technology
                  and Theoretical Computer Science, {FSTTCS} 2020, {BITS} Pilani, {K}
                  {K} Birla Goa Campus, Goa, India (Virtual Conference), December 14-18,
                  2020},
  series       = {LIPIcs},
  pages        = {37:1--37:15},
  publisher    = {Schloss Dagstuhl - Leibniz-Zentrum f{\"{u}}r Informatik},
  year         = {2020},
  url          = {https://doi.org/10.4230/LIPIcs.FSTTCS.2020.37},
  doi          = {10.4230/LIPICS.FSTTCS.2020.37},
  bibsource    = {dblp computer science bibliography, https://dblp.org}
}

@article{BlondinEFGHLMT21,
  author       = {Michael Blondin and
                  Matthias Englert and
                  Alain Finkel and
                  Stefan G{\"{o}}ller and
                  Christoph Haase and
                  Ranko Lazic and
                  Pierre McKenzie and
                  Patrick Totzke},
  title        = {The Reachability Problem for Two-Dimensional Vector Addition Systems
                  with States},
  journal      = {J. {ACM}},
  volume       = {68},
  number       = {5},
  pages        = {34:1--34:43},
  year         = {2021},
  url          = {https://doi.org/10.1145/3464794},
  doi          = {10.1145/3464794},
  bibsource    = {dblp computer science bibliography, https://dblp.org}
}

@inproceedings{CzerwinskiJ0O25,
  author       = {Wojciech Czerwinski and
                  Isma{\"{e}}l Jecker and
                  Slawomir Lasota and
                  Lukasz Orlikowski},
  editor       = {Keren Censor{-}Hillel and
                  Fabrizio Grandoni and
                  Jo{\"{e}}l Ouaknine and
                  Gabriele Puppis},
  title        = {Reachability in 3-VASS Is Elementary},
  booktitle    = {52nd International Colloquium on Automata, Languages, and Programming,
                  {ICALP} 2025, Aarhus, Denmark, July 8-11, 2025},
  series       = {LIPIcs},
  pages        = {153:1--153:20},
  publisher    = {Schloss Dagstuhl - Leibniz-Zentrum f{\"{u}}r Informatik},
  year         = {2025},
  url          = {https://doi.org/10.4230/LIPIcs.ICALP.2025.153},
  doi          = {10.4230/LIPICS.ICALP.2025.153},
  bibsource    = {dblp computer science bibliography, https://dblp.org}
}

@article{GurariI81,
  author       = {Eitan M. Gurari and
                  Oscar H. Ibarra},
  title        = {The Complexity of the Equivalence Problem for Simple Programs},
  journal      = {J. {ACM}},
  volume       = {28},
  number       = {3},
  pages        = {535--560},
  year         = {1981},
  url          = {https://doi.org/10.1145/322261.322270},
  doi          = {10.1145/322261.322270},
  bibsource    = {dblp computer science bibliography, https://dblp.org}
}

@inproceedings{BaumannDGIMSZ23,
  author       = {Pascal Baumann and
                  Flavio D'Alessandro and
                  Moses Ganardi and
                  Oscar H. Ibarra and
                  Ian McQuillan and
                  Lia Sch{\"{u}}tze and
                  Georg Zetzsche},
  editor       = {Orna Kupferman and
                  Pawel Sobocinski},
  title        = {Unboundedness Problems for Machines with Reversal-Bounded Counters},
  booktitle    = {Foundations of Software Science and Computation Structures - 26th
                  International Conference, FoSSaCS 2023, Held as Part of the European
                  Joint Conferences on Theory and Practice of Software, {ETAPS} 2023,
                  Paris, France, April 22-27, 2023, Proceedings},
  series       = {Lecture Notes in Computer Science},
  pages        = {240--264},
  publisher    = {Springer},
  year         = {2023},
  url          = {https://doi.org/10.1007/978-3-031-30829-1\_12},
  doi          = {10.1007/978-3-031-30829-1\_12},
  bibsource    = {dblp computer science bibliography, https://dblp.org}
}

@article{BlondinLMOP23,
  author       = {Michael Blondin and
                  Tim Leys and
                  Filip Mazowiecki and
                  Philip Offtermatt and
                  Guillermo A. P{\'{e}}rez},
  title        = {Continuous One-counter Automata},
  journal      = {{ACM} Trans. Comput. Log.},
  volume       = {24},
  number       = {1},
  pages        = {3:1--3:31},
  year         = {2023},
  url          = {https://doi.org/10.1145/3558549},
  doi          = {10.1145/3558549},
  bibsource    = {dblp computer science bibliography, https://dblp.org}
}

@article{Rackoff78,
  author       = {Charles Rackoff},
  title        = {The Covering and Boundedness Problems for Vector Addition Systems},
  journal      = {Theor. Comput. Sci.},
  volume       = {6},
  pages        = {223--231},
  year         = {1978},
  url          = {https://doi.org/10.1016/0304-3975(78)90036-1},
  doi          = {10.1016/0304-3975(78)90036-1},
  bibsource    = {dblp computer science bibliography, https://dblp.org}
}

@book{Vollmer,
  author       = {Heribert Vollmer},
  title        = {Introduction to Circuit Complexity - {A} Uniform Approach},
  series       = {Texts in Theoretical Computer Science. An {EATCS} Series},
  publisher    = {Springer},
  year         = {1999},
  url          = {https://doi.org/10.1007/978-3-662-03927-4},
  doi          = {10.1007/978-3-662-03927-4},
  isbn         = {978-3-540-64310-4},
  bibsource    = {dblp computer science bibliography, https://dblp.org}
}

@article{HesseAB02,
  author       = {William Hesse and
                  Eric Allender and
                  David A. Mix Barrington},
  title        = {Uniform constant-depth threshold circuits for division and iterated
                  multiplication},
  journal      = {J. Comput. Syst. Sci.},
  volume       = {65},
  number       = {4},
  pages        = {695--716},
  year         = {2002},
  url          = {https://doi.org/10.1016/S0022-0000(02)00025-9},
  doi          = {10.1016/S0022-0000(02)00025-9},
  bibsource    = {dblp computer science bibliography, https://dblp.org}
}

@inproceedings{ShakibaSZ25,
  author       = {Yousef Shakiba and
                  Henry Sinclair{-}Banks and
                  Georg Zetzsche},
  title        = {A Complexity Dichotomy for Semilinear Target Sets in Automata with
                  One Counter},
  booktitle    = {40th Annual {ACM/IEEE} Symposium on Logic in Computer Science, {LICS}
                  2025, Singapore, June 23-26, 2025},
  pages        = {594--608},
  publisher    = {{IEEE}},
  year         = {2025},
  url          = {https://doi.org/10.1109/LICS65433.2025.00051},
  doi          = {10.1109/LICS65433.2025.00051},
  bibsource    = {dblp computer science bibliography, https://dblp.org}
}

@article{tovey,
title = {A simplified NP-complete satisfiability problem},
journal = {Discrete Applied Mathematics},
volume = {8},
number = {1},
pages = {85-89},
year = {1984},
issn = {0166-218X},
doi = {https://doi.org/10.1016/0166-218X(84)90081-7},
url = {https://www.sciencedirect.com/science/article/pii/0166218X84900817},
author = {Craig A. Tovey}
}

\appendix
\section{Proofs for Section~\ref{sec:1-CVASS}}

\subsection{Proof of Theorem~\ref{thm:inte-to-qn}}\label{app-subsec:inte-to-qn}
\intetoqn*
\begin{proof}
	Let $\mach = (Q,T)$ be an instance of reachability in 1-CVASS with two configurations $p(x)$ and $q(y)$ and let $\inte$ be an interval. The same argument mentioned in Remark~\ref{rem:source-target} allows us to assume that $x \le y$. We shall now construct another 1-CVASS $\mach'$ with two configurations $p'(x')$ and $q'(y')$ such that $p(x) \act{*}_{\inte} q(y)$ in $\mach$ iff $p'(x') \actqn{*} q'(y')$ in $\mach'$.
    
    First note that if at least one of $x,y \notin \inte$, then there can be no $\inte$-run
    from $p(x)$ to $q(y)$ in $\mach$. In that case, we can simply take $(\mach',p'(x'),q'(y'))$ to be some trivial no-instance of 1-$\Reach_\qn$. Hence, we can assume that $x, y \in \inte$. The required construction of $\mach'$ is now obtained by means of a case analysis on $x,y$. To this end, we say that $x$ (resp. $y$) is an extreme point of $\inte$ if $\inte = [x,b)$ or $\inte = [x,b]$ for some $b$ (resp. if $\inte = (a,y]$ or $\inte = [a,y]$ for some $a$). Our case analysis  considers all four possibilities.

 	\subparagraph*{Case 1: $x,y$ are not extreme points of $\inte$.} 
    Therefore, there exists an $\epsilon > 0$ such that $[x-\epsilon,y+\epsilon] \subseteq \inte$. In this case, we show that $p(x) \act{*}_\inte q(y)$ is possible in $\mach$ iff
 	$p(x) \actqn{*} q(y)$ is possible in $\mach$, which then trivially gives the desired reduction.
 	The left-to-right implication is immediate. For the other side, suppose $p(x) := p_0(x_0) \actqn{\alpha_1 t_1} p_1(x_1) \actqn{\alpha_2 t_2} p_2(x_2) \dots \actqn{\alpha_n t_n} p_n(x_n) := q(y)$ is  a run over the $\qn$-semantics.
 	We say that the transition $t_i := (p_{i-1},w_i,p_i)$ is positive, negative or neutral, if $w_i > 0, w_i < 0$ or $w_i = 0$, respectively.
 	
 	Now, if all the intermediate counter values are already in $[x-\epsilon/2,y+\epsilon/2] \subseteq \inte$, then we are done. Otherwise, (since $x_0 = x_0$ and $x_n = y$), let $i$ be the first position and $j$ be the last position such that $x_{i+1}, x_{j-1} \notin [x-\epsilon/2,y+\epsilon/2]$. Hence, 
 	all the counter values up till $x_i$ and all the counter values from $x_j$ onwards belong to $[x-\epsilon/2,y+\epsilon/2]$.
 	
 	We now give an alternate run between $p_i(x_i)$ and $p_j(x_j)$ of the form
 	$p_i(x_i) \actqn{\beta_{i+1} t_{i+1}} p_{i+1}(x'_{i+1}) \actqn{\beta_{i+2} t_{i+2}} p_{i+2}(x'_{i+2}) \dots p_{j-1}(x'_{j-1}) \actqn{\beta_j t_j} p_j(x_j)$ so that all the intermediate counter values remain in $[x-\epsilon,y+\epsilon]$. Since $[x-\epsilon,y+\epsilon] \subseteq \inte$, plugging this new run between $p_i(x_i)$ and $p_j(x_j)$ into the old run will then give a $\inte$-run from $p(x)$ to $q(y)$, thereby completing the proof of this case.
 	
 	To do this, let $T_+, T_-$ and $T_0$ be the set of positive, negative and neutral transitions from the set $\{t_{i+1},\dots,t_j\}$. We now consider three further sub-cases:
 	\begin{itemize}
 		\item \textbf{Subcase 1:} $x_i = x_j$. By assumption, we have $x_j - x_i = 0 = \sum_{k=i+1}^j \alpha_k w_k$.
 		Let $W = \sum_{k=i+1}^j \alpha_k |w_k|$. We shall now modify the run so that each  transition $t_k$ is now fired by the fraction $\alpha_k/N$ for a suitably chosen $N \in \nn$ which will always allow us to stay in the desired interval
 		$[x-\epsilon,y+\epsilon]$.
 		
 		To this end, we first find an $N \in \nn$ such that $W < N \cdot \epsilon/2$. Then we set each $\beta_k = \alpha_k/N$, i.e., we fire each $t_k$ with fraction $\beta_k$. 
 		It is easy to see that that at any point in the run, the counter value lies in the range $[x_i-W/N,x_i+W/N]$.
 		Indeed, in the worst case, either all the $w_k$ are positive (resp. or all of them are negative) and in that case, the counter value increases (resp. decreases) by at most $\sum_{k=i+1}^j \beta_k |w_k| = W/N$.
 		Since $W/N < \epsilon/2$, it follows that all of the counter values along the run lie in the range
 		$ [x_i-\epsilon/2,x_i+\epsilon/2] \subseteq [x-\epsilon,y+\epsilon]$. 
 		Furthermore, it is clear that the total effect of this run is $\sum_{k=i+1}^j \beta_k w_k = \sum_{k=i+1}^j (\alpha_k/N) w_k = 0$. Hence, this indeed gives a run from $p_i(x_i)$ to $p_j(x_j)$ where all the intermediate counter values lie in $[x-\epsilon,y+\epsilon]$.
 		
 		\item \textbf{Subcase 2:} $x_i < x_j$. By assumption, we have $\Delta := x_j - x_i = \sum_{t_k \in T_+ \cup T_-
 		} \alpha_k w_k$. Let $W_+ := \sum_{t_k \in T_+} \alpha_k w_k$ and $W_- := \sum_{t_k \in T_-} \alpha_k |w_k|$. Note that $W_+ = \Delta + W_-$. We shall now modify the run so that
 		each negative and neutral transition $t_k$ is now fired by the fraction $\alpha_k/N$ (for a suitably chosen $m \in (0,1]$ ) and each positive transition $t_k$ is now fired by the fraction $\alpha_k \cdot m$. 
 		The $N,m$ will be chosen so that the constructed run will always allow us to stay in the desired interval
 		$[x-\epsilon,y+\epsilon]$.
 		
 		To this end, we first find an $N \in \nn$ such that $W_- < N \cdot \epsilon/2$ and then 
 		we find a non-zero fraction $m \in (0,1]$ such that $W_+ \cdot m = (\Delta + W_-/N)$. (The latter is possible since $W_+ = \Delta + W_- > \Delta + W_-/N$).
		Now, we set $\beta_k = \alpha_k/N$ if $k \in T_- \cup T_0$ and $\beta_k = \alpha_k \cdot m$ if $k \in T_+$
		By the same argument as the previous subcase, at any point in the run, the counter value is 
		in the range $[x_i-\epsilon/2,x_j+\epsilon/2] \subseteq [x-\epsilon,y+\epsilon]$
		and the total effect of the run is $W_+ \cdot m - W_-/N = \Delta$. Hence, this indeed gives a run from $p_i(x_i)$ to $p_j(x_j)$ where all the intermediate counter values lie in $[x-\epsilon,y+\epsilon]$.
		
		\item \textbf{Subcase 3:} $x_i > x_j$. This is similar to subcase 2, except we reverse the roles of $T_+$ and $T_-$. More precisely, we have $\Delta := x_i - x_j = \sum_{t_k \in T_+ \cup T_- } \alpha_k v_k$.
		Let $W_+ := \sum_{t_k \in T_+} \alpha_k w_k$ and $W_- := \sum_{t_k \in T_-} \alpha_k |w_k|$. Note that $W_+ = \Delta + W_-$. 
		
		We first find an $N \in \nn$ such that $W_+ < N \cdot \epsilon/2$ and then 
		we find a non-zero fraction $m \in (0,1]$ such that $W_- \cdot m = (\Delta + W_+/N)$.
		Now, we set $\beta_k = \alpha_k/N$ if $k \in T_+ \cup T_0$ and $\beta_k = \alpha_k \cdot m$ if $k \in T_-$.
		Similar to the previous subcase, this gives the required run.
 	\end{itemize}
 
	This completes the proof of all the three subcases and also the proof of this case.
	
	\subparagraph*{Case 2: $x$ is an extreme point, $y$ is not.} Therefore, $\inte$ is of the form $[x,b)$ or $[x,b]$
	with $y < b$. This means that in any $\inte$-run from $p(x)$ and $q(y)$ in $\mach$, the first non-neutral
	transition must be a positive transition, as otherwise the counter value would have to leave the interval $\inte$.
	Hence, any $\inte$-run from $p(x)$ to $q(y)$ in $\mach$ must necessarily be of the form
	$p(x) \act{*}_\qn d(x) \act{}_\qn e(z) \act{*}_\inte q(y)$ where $z > x$ and the 
	run between $p(x)$ and $d(x)$ consists only of neutral transitions. Furthermore, since $z > x$, it follows that there is an $\epsilon > 0 $ such that $[z-\epsilon,y+\epsilon] \subseteq \inte$. Applying the previous case to the run
	$e(z) \act{*}_\inte q(y)$, we can conclude that this can happen iff $e(z) \actqn{*} q(y)$. Hence, $p(x) \act{*}_\inte q(y)$ is possible iff there is a run of the form $p(x) \act{*}_\qn d(x) \act{}_\qn e(z) \act{*}_\qn q(y)$ where $z > x$ and the run between $p(x)$ and $d(x)$ consists only of neutral transitions.
	We now construct another 1-CVASS $\mach'$ that explicitly tracks runs of this form.
	
	The 1-CVASS $\mach'$ will have all the states of $\mach$ as well as another copy $\{d' : d \in Q\}$.
	The first copy will track all the neutral transitions fired from $p(x)$ and the moment a positive transition is fired, it is simulated by moving to the second copy of states, where $\mach$ is simulated freely.
	Formally, the transitions of $\mach'$ are as follows: Suppose $(d,w,e) \in \delta$ is a transition of $\mach$.
	Then, corresponding to it $\mach'$ has the following transitions:
	$(d,w,e)$ if $w = 0$, $(d,w,e')$ if $w > 0$ and $(d',w,e')$. 
	
	From the construction of $\mach'$, it is easy to verify that
	$p(x) \act{*}_\qn d(x) \act{}_\qn e(z) \act{*}_\qn q(y)$ is a $\qn$-run in $\mach$ with $z > x$ and the 
	run between $p(x)$ and $d(x)$ consisting only of neutral transitions iff 
	$p(x) \actqn{*} q'(y)$ is a $\qn$-run in $\mach'$. By the above discussion 
	this allows us to conclude that $p(x) \act{*}_\inte q(y)$ is an $\inte$-run in $\mach$ iff
	$p(x) \actqn{*} q'(y)$ is a $\qn$-run in $\mach'$.
	
    \subparagraph*{Case 3: $x$ is not an extreme point, $y$ is. } Therefore, $\inte$ is of the form $(a,y]$ or $[a,y]$
	with $a < x$. This means that in any $\inte$-run from $p(x)$ and $q(y)$ in $\mach$, the last non-neutral
	transition must be a positive transition, as otherwise the counter value would have to have left the interval $\inte$ to reach $y$. Similar to the previous case,
	we can show that $p(x) \act{*}_\inte q(y)$ is possible in $\mach$ iff
	 iff there is a run of the form $p(x) \act{*}_\qn d(z) \act{}_\qn e(y) \act{*}_\qn q(y)$ where $z < y$ and the run between $e(y)$ and $q(y)$ consists only of neutral transitions.

	Hence, the construction for this case is similar to the previous case, except we have to track the last non-neutral fired transitions instead of the first non-neutral fired transition. More precisely, we let $\mach'$ have all the states of $\mach$ as well as another copy $\{d' : d \in Q\}$. Furthermore, its transitions are as follows: Suppose $(d,w,e) \in \delta$ is a transition of $\mach$. Then, corresponding to it $\mach'$ has the following transitions: $(d,w,e)$, 
	$(d,w,e')$ if $w > 0$ and $(d',w,e')$ if $w = 0$. 
	Similar to the previous case, we can show that $p(x) \act{*}_\inte q(y)$ is an $\inte$-run in $\mach$ iff $p(x) \actqn{*} q'(y)$ is a  $\qn$-run in $\mach'$.
	
	\subparagraph*{Case 4: $x,y$ are both extreme points. } Therefore, $\inte = [x,y]$.
	If $x = y$, it follows that any $\inte$-run between $p(x)$ and $q(y)$ must use only neutral transitions of $\mach$, and so we simply delete all the positive and negative transitions of $\mach$ to get $\mach'$.
	
	On the other hand if $x \neq y$, this case is a combination of the above two cases and so we have to track both the first and last fired non-neutral transitions. So, we have to simply combine the constructions from the above two cases. More precisely, we let $\mach'$ have all the states of $\mach$ as well as two other copies $\{d' , d'' : d \in Q\}$. Furthermore, its transitions are as follows:  Suppose $(d,w,e) \in \delta$ is a transition of $\mach$.
	Then, corresponding to it $\mach'$ has the following transitions: $(d',w,e')$, 
	$(d,w,e')$, if $w > 0$, $(d',w,e'')$ if $w > 0$ and $(d,w,e), (d'',w,e'')$ if $w = 0$.
	Similar to the previous two cases, it can be shown that 
	$p(x) \act{*}_\inte q(y)$ is an $\inte$-run in $\mach$ iff
	$p(x) \actqn{*} q''(y)$ is a $\qn$-run in $\mach'$.
\end{proof}

\subsection{Proof of Theorem~\ref{thm:charac-equality}}\label{app-subsec:charac-equality}
\characequality*
\begin{proof}	
	Suppose $p(x) := p_0(x_0) \actqn{\alpha_1 t_1} p_1(x_2) \actqn{\alpha_2 t_2} \dots p_{n-1}(x_{n-1}) \actqn{\alpha_n t_n} p_n(x_n) := q(x)$ is a $\qn$-run where each $t_i = (p_{i-1},w_i,p_i)$.
	This gives rise to a corresponding path $P$ in $\graph_\mach$ given by $p := p_0 \act{w_1} p_1 \act{w_2} \dots \act{w_n} p_n := q$. Now, if $T_{+}(P) \neq \emptyset$ and $T_{-}(P) = \emptyset$ (or vice versa),
	it would follow that $\sum_{i=1}^n \alpha_i w_i > 0$ (or $\sum_{i=1}^n \alpha_i w_i < 0$, respectively), which would be a contradiction. Hence, $P$ is a desired path in $\graph_\mach$.
	
	Suppose, $P$ is a path in $\graph_\mach$ from $p$ to $q$
	such that either $T_+(P), T_-(P) \neq \emptyset$ or $T_+(P) = T_-(P) = \emptyset$.
	Let $P = p:= p_0 \act{w_1} p_1 \act{w_2} \dots p_{n-1} \act{w_n} p_n := q$.
	Correspondingly, we have the transitions $t_i := (p_{i-1},w_i,p_i$ in $\mach$ for each $i$. By assumption, either both $T_+(P)$ and $T_{-}(P)$ are empty or both of them are non-empty.
	Suppose both of them are empty. Then it is clear that $w_i = 0$ for each $i$ and so if we fire each $t_i$ with any arbitrary non-zero fraction, we have a $\qn$-run from $p(x)$ to $q(x)$ in $\mach$.
	On the other hand, suppose both $T_+(P)$ and $T_{-}(P)$ are non-empty. In this case, we fire each $t_i$ with a fraction $\alpha_i$ as follows: If $i \in T_+(P)$ then we set $\alpha_i := \frac{1}{w_i |T_{+}(P)|}$; if $i \in T_{-}(P)$ then we set $\alpha_i := \frac{-1}{w_i|T_{-}(P)|}$; if $i \in T_0(P)$ then we set $\alpha_i$ to be any arbitrary fraction. By construction, the
	net effect of this run on the counter is exactly $|T_{+}(P)| \cdot \frac{1}{|T_{+}(P)|} - |T_{-}(P)| \cdot \frac{1}{|T_{-}(P)|} = 0$. Hence, this gives a $\qn$-run from $p(x)$ to $q(x)$.
\end{proof}

\subsection{Proof of Proposition~\ref{prop:pumpable-cycle}}\label{app-subsec:pumpable-cycle}
\pumpablecycle*
\begin{proof}
	Let $C$ be a pumpable cycle given by $r := r_0 \act{w_1} r_1 \act{w_2} \dots r_{n-1} \act{w_n} r_n := r$. 
	Let $t_i := (r_{i-1},w_i,r_i)$ be the transitions in $\mach$ corresponding to this path.	
	We will now fire each $t_i$ with some appropriately-chosen fraction $\alpha_i$ so that the net effect 
	of the fired run on the counter will be some non-zero fraction $\Delta := \frac{y-x}{m}$ for some $m \in \nn$. Intuitively, we will do this, by firing each
	negative transition corresponding to $T_-(C)$ with a very tiny fraction so that their combined effect is of the form $-\epsilon/m$ for a small $\epsilon$ and large $m$. We then fire each positive transition corresponding to $T_+(C)$ with an appropriately chosen fraction so that their
	combined effect is exactly equal to $(v-u+\epsilon)/m$. 
	This enables us to construct a run from $r(u)$ to $r(u+\Delta)$. 
	Repeating this same run over and over, we can reach 
	$r(u+2\Delta), r(u+3\Delta)$ and so on till $r(u+m\Delta) := r(v)$. We now formalize this idea.

	Let $W_+(C) := \sum_{w_i > 0} w_i$ and $W_-(C) := \sum_{w_i < 0} |w_i|$. By assumption, $W_+(C) > 0$. 
	If $W_-(C) = 0$, then we let $\epsilon = 0$; otherwise, we let $\epsilon \in (0,1]$ such that $\epsilon \le W_-(C)$. We then let $\delta := v - u + \epsilon$ with $m = \lceil \delta \rceil$.
	We now fire each $t_i$ with a fraction $\alpha_i$ as follows: If $w_i \ge 0$, then we set 
	$\alpha_i := \frac{\delta}{m W_{+}(C)}$; otherwise, we set $\alpha_i := \frac{\epsilon}{m W_{-}(C)}$. 
	By construction, it is seen that this is a $\qn$-run from $r(u)$ to $r(u+\Delta)$
	where $\Delta = \sum_{w_i > 0} \frac{v - u + \epsilon}{m W_{+}(C)} \cdot w_i - \sum_{w_i < 0} \frac{\epsilon}{m W_{-}(C)} \cdot |w_i| = \frac{v-u+\epsilon}{m} - \frac{\epsilon}{m} = \frac{v-u}{m}$.
	Repeating this same run from $r(u+\Delta)$ we can then reach $r(u+2\Delta)$, then $r(u+3\Delta)$ and so on,
	till $r(u+m\Delta) = r(v)$.
\end{proof}

\subsection{Proof of Theorem~\ref{thm:charac-inequality}}\label{app-subsec:charac-inequality}
\characinequality*
\begin{proof}
	Let us first prove the left-to-right implication. Suppose $p(x) := p_0(x_0) \actqn{\alpha_1 t_1} p_1(x_2) \actqn{\alpha_2 t_2} \dots p_{n-1}(x_{n-1}) \actqn{\alpha_n t_n} p_n(x_n) := q(y)$ is a $\qn$-run where each $t_i = (p_{i-1},w_i,p_i)$. This gives rise to a corresponding path $P$ in $\graph_\mach$ given by $p := p_0 \act{w_1} p_1 \act{w_2} \dots \act{w_n} p_n := q$. If there is a pumping cycle (at some state $p_i$), then we set $I$ to be the path from $p_0$ to $p_i$, $C$ to be the pumping cycle from $p_i$ to $p_i$ and $F$ to be the path from $p_i$ to $p_n$. This gives rise to a knot $K = (I,C,F)$ of $\graph_\mach$. Hence, in this case, we are done.
	
	Suppose there is no pumping cycle at any of the states $p_i$. Then, any cycle that might appear in $P$ must not have any strictly positive labels. Let $P'$ be the path obtained from $P$ by removing all of its cycles. 
	$P'$ is a simple path from $p$ to $q$ in $\graph_\mach$. Furthermore, since we obtained $P'$ by only removing 
    edges from $P$ with non-positive labels, it follows that $W(P') \ge W(P) = \sum_{1 \le i \le n} w_i \ge \sum_{1 \le i \le n} \alpha_i w_i = y-x$. Hence, in this case too, we are done.

    Let us now prove the right-to-left implication. If there is a simple path $P$ from $p$ to $q$ in $\graph_\mach$
    such that $W(P) \ge y-x$, then by Proposition~\ref{prop:simple-path}, we are done.
    Otherwise, suppose $K = (I,C,F)$ is a knot from $p$ to $q$ in $\graph_\mach$.
    Let $r$ be the state at which the cycle $C$ begins and ends.
    Since $x < y$, we can pick some $\epsilon, \delta \in (0,1]$ such that $x+\epsilon W(I) < y-\delta W(F)$. Now, by Proposition~\ref{prop:simple-path} we have a
    $\qn$-run from $p(x)$ to $r(x+\epsilon W(I))$. By Proposition~\ref{prop:pumpable-cycle}, we have
    a $\qn$-run from $r(x+\epsilon W(I))$ to $r(y - \delta W(F))$. Finally, once again by Proposition~\ref{prop:simple-path}, we have a path from $r(y-\delta W(F))$ to $q(y)$. Stringing together these three runs, we get a $\qn$-run from $p(x)$ to $q(y)$ and so in this case too, we are done.
\end{proof}

\subsection{Proof of Theorem~\ref{lem:testing-conditions}}\label{app-subsec:testing-conditions}

\testingconditions*
\begin{proof}
    We begin by recalling these three conditions, which we have to test in $\graph_\mach$:
    \begin{itemize}
    \item Condition 1: $x = y$ and there is a path $P$ in $\graph_\mach$ from $p$ to $q$ such that either $T_+(P), T_-(P) \neq \emptyset$ or $T_+(P) = T_-(P) = \emptyset$.
    \item Condition 2:  $x < y$ and there is a knot $K$ in $\graph_\mach$ from $p$ to $q$.
    \item Condition 3: $x < y$ and there is a path $P$ of length at most $|Q|$ in $\graph_\mach$ from $p$ to $q$ such that $W(P) \ge y-x$.
\end{itemize}

\paragraph*{Testing Conditions 1 and 2}  We first describe our $\NL$ algorithms for Conditions 1 and 2. For both of these conditions, we shall show that whenever a desired path or knot exists, there is one of size at most $4|Q|$.

Indeed, suppose there is a path $P$ in $\graph_\mach$ from $p$ to $q$ such that either $T_+(P), T_-(P) \neq \emptyset$ or $T_+(P) = T_-(P) = \emptyset$. Suppose the latter condition is true. Then, we can safely remove any cycle from $P$ whilst still maintaining the property that both the set of positive and negative edges along the path is empty. Hence, this way, we can shrink the length of $P$ until it becomes at most $|Q|$.

Similarly, suppose the former condition is true. Without loss of generality, we can assume that a transition from $T_+(P)$ is encountered first in $P$ before any transition from $T_-(P)$ (The proof for the other case is similar). Let $(p',w',q')$ and $(p'',w'',q'')$ be the first occurrence of an edge from $T_+(P)$ and $T_-(P)$ along the path $P$. 
Then, we can safely remove any cycle in $P$ that is i) either between $p$ and $p'$, ii) either between $q'$ and $p''$ or iii) either between $q''$ and $q$. Hence, this way, we can shrink the length of $P$ until it becomes at most $3|Q|$.

The above argument immediately gives an $\NL$ algorithm for testing Condition 1: Simply guess a path of length at most $3|Q|$ from $p$ in a step-by-step manner and accept if this path ends at $q$ and either we encountered both a positive and a negative edge along the path or we encountered neither along the path. This completes the algorithm for Condition 1.

For testing Condition 2, suppose $K = (I,C,F)$ is a knot from $p$ to $q$ in $\graph_\mach$. 
Let $r$ be the state at which the cycle $C$ begins and ends. By definition of a knot, we can safely remove
any cycles in $I$ and $F$ and we would still be left with a knot. Hence, this way, we can shrink the lengths of $I, F$ until they become at most $|Q|$. Similarly, for the cycle $C$, let $(p',w',q')$ be the first occurrence of an edge from $T_+(C)$. Then, we can safely remove any cycle in $C$ that is i) either between $r$ and $p'$ or ii) between $q'$ and $r$. 
Hence, this way, we can shrink the length of $C$ until it becomes at most $2|Q|$. Overall,
this means that the knot $K$ can be taken to be of length at most $4|Q|$. Similar to the previous condition, this immediately gives an $\NL$ algorithm for testing Condition 2.

\paragraph*{Testing Condition 3} We now show how to test condition 3. First, suppose the 
update vectors are over the integers and unary encoding is used. Then, to test Condition 3, we simply need to guess a path $P$ of length at most $|Q|$ in a step-by-step manner, maintain the sum of the labels encountered along the path and then accept if this sum is ultimately $\ge y-x$. Since
all the update vectors are integers encoded in unary, the sum of the labels at each step can be stored in logarithmic space (by maintaining this sum in binary). This allows us to do the entire computation in $\NL$.

Now, suppose the update vectors are integers and binary encoding is used. In this case,
we first make the graph $\graph_\mach$ acyclic as follows: We expand the set of vertices to $\{(d,i) : d \in Q, 0 \le i \le n\}$. Then, for every edge $(d,w,e)$ of $\graph_\mach$, we will have edges $((d,i),w,(e,i+1))$ for every $0 \le i \le n-1$. We will also have the edges $((d,i),0,(d,i+1))$ for every $d \in Q, 0 \le i \le n-1$. It is immediately clear that there is a path $P$ in $\graph_\mach$ from $p$ to $q$ of length at most $|Q|$ with $W(P) \ge y-x$ iff there is a path $P'$ in this new graph $\graph'$ from $(p,0)$ to $(q,n)$ with $W(P') \ge y-x$. The latter problem is precisely the problem of coverability in 1-dimensional integer VASS, i.e., 1-CVASS over the $\qn$-semantics where all the fractions have to be 1. This latter problem is known to be in $\AC^1$~\cite[Theorem III.2]{ShakibaSZ25}. Since $\graph'$ can be constructed from $\graph_\mach$ in logarithmic space and since $\mathsf{L} \subseteq \NL \subseteq \AC^1$, it follows that we have an $\AC^1$ algorithm in this case.

Finally, suppose the update vectors are rational numbers. We show that we can reduce this case to the previous case of integer update vectors over the binary encoding. Indeed, given the graph $\graph_\mach$, we first multiply all of its labels with the produce $B$ of all of its denominators, as well as the denominators of $x,y$. It immediately follows that there is a path $P$ in $\graph_\mach$ from $p$ to $q$ of length at most $|Q|$ with $W(P) \ge y-x$ iff there is a path $P'$ in this new graph $\graph'$ of length at most $|Q$ from $p$ to $q$ with $W(P') \ge B(y-x)$. Furthermore, $\graph'$ only has integer labels
over its edges and so by the previous paragraph, we can solve this in $\AC^1$.
Since the product of a given collection of numbers
can be computed in $\AC^1$~\cite[Corollary 4.1]{HesseAB02}, this proves that the third condition can be tested in $\AC^1$ in this case as well. This completes the description of the algorithms for Condition 3.
\end{proof}

\section{Proofs of Section~\ref{sec-2CVASS}}

%

\subsection{Proof of Proposition~\ref{proposition-restricted-reachability-clause-gadget}}
\label{app-proposition-restricted-reachability-clause-gadget-proof}
\propositionrestrictedreachabilityclausegadget*
\begin{proof}
	\(\Rightarrow\): Suppose there is a restricted $\qn$-run from $b_0(f(A))$ to $b_m(0)$. By construction of $\mach_C$, in the $i^{th}$ step along this run the counter is decremented by  $\sum_{\ell \in S} \frac{1}{P(\ell)}$ for some non-empty subset $S\subseteq \{\ell_i^1, \ell_i^2, \ell_i^3\}$. This implies
	that $f(A) = \sum_{\ell \in S'} \frac{f_\ell}{P(\ell)}$ for some subset $S' \subseteq \{x_1,\dots,x_n,\overline{x}_1,\dots,\overline{x}_n\}$ such that $S'$ contains at least one literal from every clause and each $f_\ell \le \#(\ell) \le 4$.
	By Proposition~\ref{prop:Egyptian-prime-encoding},  $S'$ is precisely the set of literals set to true by $A$, which implies that $A$ satisfies $\varphi$.
	
	\(\Leftarrow\): Suppose $A$ satisfies the formula $\varphi$. Then, we can construct a restricted $\qn$-run from $b_0(f(A))$
	to $b_m(0)$ as follows: Initially we begin at the configuration $b_0(f(A)) := b_0(f(A)$. 
	At any point, if we are at $b_i(k_i)$
	for some $0 \le i < m$, we do the following: We consider the non-empty subset $S$ of literals $\ell_i^1, \ell_i^2, \ell_i^3$ that are satisfied by $A$ and we pick the transition from $b_i$ to $b_{i+1}$ that decrements the counter by $\sum_{\ell \in S} \frac{1}{P(\ell)}$. 
	Note that along the run the counter is decremented exactly by $\sum_{\ell : A(\ell) \text{ is true} } \frac{\#(\ell)}{P(\ell)} = f(A)$ and so the counter will be exactly 0 when the state $b_m$ is reached.
\end{proof}

\subsection{Proof of Lemma~\ref{lemma-coverability-2cvass-npcomplete}}
\label{app-lemma-coverability-2cvass-npcomplete-proof}
\lemmacoverabilitycvassnpcomplete*
\begin{proof}
	We will prove the claim for the case of $\qnz$-runs; the proof for the case of $\qn$-runs is similar (and in fact easier).
	
	\(\Rightarrow \): Suppose there is a restricted $\qnz$-run from $a_0(0)$ to $b_m(0)$ in $\mach$ given by
	$$a_0(0) := a_0(x_0) \act{t_1} a_1(x_1) \dots \act{t_n} a_n(x_n) \act{s_1} b_1(y_1) \dots \act{s_m} b_m(y_m) := b_m(0)$$ By the definition of $K$ it follows that each $x_i, y_j \le 3K$. Hence, if we fire the exact same run in $\mach'$, we will get a $\qnz$-run of the form 	\begin{multline*}
		a_0(0,3K) := a_0(x_0,3K-x_0) \act{t_1} a_1(x_1,3K-x_1+1) \dots \act{t_n} a_n(x_n,3K-x_n+n) \\\act{s_1} b_1(y_1,3K-y_1+n+1) \dots \act{s_m} b_m(y_m,3K-y_m+K) := b_m(0,4K)	
	\end{multline*}

	\(\Leftarrow \): Suppose there is a $\qnz$-run from $a_0(0,3K)$ covering $b_m(0,4K)$ in $\mach'$. Let the run be given by	\begin{multline*}a_0(0,3K) := a_0(x_0,x_0') \act{\alpha_1 t_1} a_1(x_1,x_1') \dots \act{\alpha_n t_n} a_n(x_n,x_n') \\\act{\beta_1 s_1} b_1(y_1,y_1') \dots \act{\beta_m s_m} b_m(y_m,y_m')\end{multline*} 	with $y_m \ge 0$ and $ y_m' \ge 4K$. Notice that for each $i$, we have that $$x_i + x_i' = 3K + \sum_{j=1}^i \alpha_j$$ and $$y_i + y_i' = 3K + \sum_{j=1}^n \alpha_j + \sum_{j=1}^i \beta_j$$ Hence, $$0+ 4K \le y_m + y_m' = 3K + \sum_{j=1}^n \alpha_j + \sum_{j=1}^m \beta_j \le 3K + (n+m) = 4K$$ Therefore
	$y_m + y_m' = 4K$ and $\sum_{j=1}^n \alpha_j + \sum_{j=1}^m \beta_j = K$, which can happen
	iff $y_m = 0, y_m' = 4K$ and each $\alpha_j, \beta_j$ is 1.
	This then immediately implies that we have a restricted $\qnz$-run $a_0(0) := a_0(x_0) \act{t_1} a_1(x_1) \dots \act{t_n} a_n(x_n) \act{s_1} b_1(y_1) \dots \act{s_m} b_m(y_m) := b_m(0)$ in $\mach$.
\end{proof}
\section{Proofs of Section~\ref{sec-3CVASS}}

\subsection{Proof of Proposition~\ref{prop:assignment-gadget}}
\label{app-prop:assignment-gadget-proof}
\assignmentgadget*
\begin{proof}
	By construction a $\qnz$-run from $a_i(g_i,h_i)$ to $a_{i+1}(g_{i+1},h_{i+1})$ is possible iff it is either of the form $$a_i(g_i,h_i) \act{\alpha (-P(x_i)^{\#(x_i)},1)} a_i'(g_i-\alpha P(x_i)^{\#(x_i)}, h_i+\alpha) \act{\beta (1,-1)}  a_{i+1}(g_i-\alpha P(x_i)^{\#(x_i)}+ \beta, h_i+\alpha-\beta)$$ or of the form 
	$$a_i(g_i,h_i) \act{\alpha (-P(\overline{x}_i)^{\#(\overline{x}_i)},1)} a_i'(g_i-\alpha P(\overline{x}_i)^{\#(\overline{x}_i)}, h_i+\alpha) \act{\beta (1,-1)} a_{i+1}(g_i-\alpha P(\overline{x}_i)^{\#(\overline{x}_i)}\alpha + \beta,  h_i+\alpha-\beta)$$

	Since $Z(a_i) = Z(a_{i+1}) = \{2\}$ and $Z(a_i') = \{1\}$, it follows that the first case can happen iff $h_i = 0, h_i+\alpha - \beta = 0$ and $g_i-\alpha P(x_i)^{\#(x_i)} = 0$. This is true iff 
	$\alpha = \beta = g_i/P(x_i)^{\#(x_i)}$ and so $h_i = h_{i+1} = 0$ and $g_{i+1} = \beta = g_i/P(x_i)^{\#(x_i)}$.
	Analogously, we can show that the second case can happen iff $h_i = h_{i+1} = 0$
	and $g_{i+1} = g_i/P(\overline{x}_i)^{\#(\overline{x}_i)}$.
\end{proof}

	\subsection{Proof of Proposition~\ref{prop:clause-gadget}}
\label{app-prop:clause-gadget-proof}
\clausegadget*
\begin{proof}
	By construction a $\qnz$-run from $b_i(g_i,h_i)$ to $b_{i+1}(g_{i+1},h_{i+1})$ is possible iff it is of the form $$b_i(g_i,h_i) \act{\alpha(-1,1)} b_i'(g_i-\alpha, h_i+\alpha) \act{\beta(p,-1)} b_{i+1}(g_i-\alpha + p\beta, h_i+\alpha-\beta)$$ where $p$ is the product of some non-empty
	subset of $\{P(\ell_i^1),P(\ell_i^2),P(\ell_i^3)\}$.
	
	Since $Z(b_i) = Z(b_{i+1}) = \{2\}$ and $Z(b_i') = \{1\}$, it follows that this can happen iff $h_i = 0, h_i+\alpha - \beta = 0$ and $g_i-\alpha = 0$. This is true iff $\alpha = \beta = g_i$ and so
	$h_i = h_{i+1} = 0, g_i = \alpha \le 1$ and $g_{i+1} = p \beta = g_i \cdot p$.
\end{proof}

\subsection{Proof of Lemma~\ref{lemma-clauseapplication}}\label{app-lemma-clauseapplication-proof}
\clauseapplication*
\begin{proof}
	\(\Leftarrow \): Suppose $A$ satisfies $\varphi$. Then we can construct a $\qnz$-run from $b_0(f(A),0)$ to $b_m(1,0)$ as follows:
	Initially we begin at $b_0(f(A),0)$. 
	At any point, if we are at $b_i(g_i,0)$ for some $0 \le i < m$, 
	we do the following: We consider the non-empty subset $S$ of literals in $\ell_i^1,\ell_i^2,\ell_i^3$
	that are satisfied by $A$ and we let $p = \prod_{\ell \in S} P(\ell)$. By Proposition~\ref{prop:clause-gadget}, there is a $\qnz$-run from $b_i(g_i,0)$ to $b_{i+1}(g_i \cdot p, 0)$ and we extend the run from $b_i$ accordingly. 
	By construction, it is clear that, during the entire run from $b_0$ to $b_m$, we multiply the first counter by exactly $\prod_{\ell: A(\ell) \text{ is true} } P(\ell)^{\#(\ell)} = 1/f(A)$. Hence, the final value of the first counter will be exactly 1 and so we reach $b_m(1,0)$ at the end of the run.
	
\(\Rightarrow\):	Suppose there is a $\qnz$-run from $b_0(f(A),0)$ to a configuration $b_m(g,h)$ such that $g \ge 1, h \ge 0$. This run must be of the form
	$$b_0(g_0,h_0) \actqnz{*} b_1(g_1,h_1) \dots \actqnz{*} b_m(g_m,h_m) := b_m(1,0) $$	
	
	Repeated applications of Proposition~\ref{prop:clause-gadget} imply that, for each $1 \le i \le m$, there is a number $j_i$ which is the product of some non-empty
	subset of $\{P(\ell_i^1),P(\ell_i^2),P(\ell_i^3)\}$ such that $g_i = j_i \cdot g_{i-1}$
	This implies that, $g_m := 1 =  g_0 \cdot \prod_{1 \le i \le m} j_i = f(A) \cdot \prod_{1 \le i \le m} j_i$. By definition, this can happen iff for each $1 \le i \le m$, $1/f(A)$ is divisible by 
	at least one of $P(\ell_i^1), P(\ell_i^2), P(\ell_i^3)$. By definition of $f(A)$, this implies that 
	$A$ sets at least one of $\ell_i^1, \ell_i^2, \ell_i^3$ to true, thereby proving that $A$ satisfies $\varphi$.
\end{proof}

	\subsection{Proof of Lemma~\ref{lemma-coverability}}\label{app-lemma-coverability-proof}
\coverability*
\begin{proof}
	\(\Leftarrow\): Suppose $A$ satisfies $\varphi$. By Lemma~\ref{lemma-clauseapplication}, there is a $\qnz$-run from $b_0(f(A),0)$ to
	$b_m(1,0)$. If we then fire the transition $(b_m,(-1,1),b_m')$ with fraction 1, we can reach $b_m'(0,1)$.
	
\(\Rightarrow \):	Suppose there is a $\qnz$-run from $b_0(f(A),0)$ covering $b_m'(0,1)$. The last step in this run
	must be of the form $b_m(g,h) \act{\alpha (-1,1)} b_m'(g',h')$ where $\alpha \in (0,1], g' \ge 0, h' \ge 1$.
	Since $Z(b_m) = \{2\}$ and $Z(b_m') = \{1\}$, it follows that $h = g' = 0$.
	Hence, we have $g' := 0 = g - \alpha, h' = h + \alpha = \alpha$.
	Since $h' \ge 1$ by assumption, $\alpha = 1$ and so $h' = g = 1$ as well.
	Hence, if we consider the prefix of this run until the last step, we actually have a $\qnz$-run from $b_0(f(A),0)$ to $b_m(1,0)$. By Lemma~\ref{lemma-clauseapplication}, $A$ must satisfy $\varphi$.
\end{proof}

\subsection{Proof of Proposition~\ref{prop:cc}}\label{app-prop:cc-proof}
\cc*
\begin{proof}
	Remember that since we consider only \(\qnz \)-semantics, for a $\qnz$-run of the form Eq.~\ref{eq:run-without-zero-tests},  the values $h_i, r_i, g_i', k_i' \geq  0$ for all \(i\). Hence all of them are 0 iff their combined sum is 0, i.e., a $\qnz$-run of the form Eq.~\ref{eq:run-without-zero-tests} is a good run iff 
	\begin{equation}\tag{\ref{eq:control}} 
		\sum_{i=0}^{n-1} g_i'  + \sum_{i=1}^n h_i + \sum_{i=0}^m k_i' + \sum_{i=1}^m r_i = 0   
	\end{equation}
	
	Let us now try and expand out each summand in this equation. By definition, for each $0 \le i \le n-1$, $g_i'$ is the effect of the run on the first counter until the point the state $a_i'$ is reached. 
	This means that
	\begin{equation*}
		g_i' = 1 + \sum_{j=0}^{i-1} (\alpha_j v_j[1] + \alpha_j' v_j'[1]) + \alpha_i v_i[1]   
	\end{equation*}
	Hence $\sum_{i=0}^{n-1} g_i' = \sum_{i=0}^{n-1} \big(1 + \sum_{j=0}^{i-1} (\alpha_j v_j[1] + \alpha_j' v_j'[1]) + \alpha_i v_i[1]\big)$. Simplifying this sum by bunching together common terms, we get
	\begin{equation}\label{eq:g_i}
		\sum_{i=0}^{n-1} g_i' = n + \sum_{i=0}^{n-1} (n-i) \alpha_i v_i[1] + \sum_{i=0}^{n-2} (n-1-i) \alpha_i' v_i'[1]    
	\end{equation}
	
	Similarly, each $h_i, k_i', r_i$ satisfies
	\begin{equation*}
		h_i = \sum_{j=0}^{i-1} (\alpha_j v_j[2] + \alpha_j' v_j'[2])   
	\end{equation*}
	\begin{equation*}
		k_i' = 1 + \sum_{j=0}^{n-1} (\alpha_j v_j[1] + \alpha_j' v_j'[1]) + 
		\sum_{j=0}^{i-1} (\beta_j w_j[1] + \beta_j' w_j'[1]) + \beta_i w_i[1]   
	\end{equation*}
	\begin{equation*}
		r_i =\sum_{j=0}^{n-1} (\alpha_j v_j[2] + \alpha_j' v_j'[2]) + 
		\sum_{j=0}^{i-1} (\beta_j w_j[2] + \beta_j' w_j'[2])
	\end{equation*}
	
	Hence, similar to the $g_i'$ we get,
	\begin{equation}\label{eq:h_i}
		\sum_{i=1}^n h_i = \sum_{i=0}^{n-1} (n-i) \alpha_i v_i[2] + \sum_{i=0}^{n-1} (n-i) \alpha_i' v_i'[2]  
	\end{equation}
	\begin{equation}\label{eq:k_i}
		\sum_{i=0}^{m} k_i' = (m+1) + \big((m+1) \cdot \sum_{i=0}^{n-1} (\alpha_i v_i[1] + \alpha_i' v_i'[1])\big) + 
		\sum_{i=0}^{m} (m-i+1) \beta_i w_i[1] + \sum_{i=0}^{m-1} (m-i) \beta_i' w_i'[1]    
	\end{equation}
	\begin{equation}\label{eq:r_i}
		\sum_{i=1}^{m} r_i = \big(m \cdot \sum_{i=0}^{n-1} (\alpha_i v_i[2] + \alpha_i' v_i'[2])\big) + 
		\sum_{i=0}^{m-1} (m-i) \beta_i w_i[2] + \sum_{i=0}^{m-1} (m-i) \beta_i' w_i'[2]    
	\end{equation}
	
	Plugging in Equations~\ref{eq:g_i},~\ref{eq:h_i},~\ref{eq:k_i}~and~\ref{eq:r_i} into Equation~\ref{eq:control}
	and bunching together common terms, we get that a $\qnz$-run of the form Eq. ~\ref{eq:run-without-zero-tests} is a good $\qnz$-run iff
	\begin{multline*}\tag{\ref{eq-tester}} 
		n + (m+1) + \left(\sum_{i=0}^{n-1} \alpha_i \big((n-i+m) (v_i[1] + v_i[2]) + v_i[1]\big) \right) + 
		\left(\sum_{i=0}^{n-1} \alpha_i' (n-i+m) (v_i'[1] + v_i'[2]) \right) + \\  
		\left(\sum_{i=0}^{m} \beta_i \big((m-i) (w_i[1] + w_i[2]) + w_i[1]\big) \right) +
		\left(\sum_{i=0}^{m-1} \beta_i' (m-i) (w_i'[1] + w_i'[2]) \right) = 0
	\end{multline*}
\end{proof}

\subsection{Proof of Lemma~\ref{lem:satisfiable-with-controlling-counter}}\label{app-lem:satisfiable-with-controlling-counter-proof}

We first begin by formally constructing $\mach'$. $\mach'$ is constructed from $\mach$ by adding a third counter to $\overline{\mach}$, keeping all of its states and changing its transitions as follows:
\begin{itemize}
    \item The transition $(a_i,v_i,a_i')$ is modified as $(a_i,(v_i,-(n-i+m) (v_i[1] + v_i[2]) - v_i[1]), a_i')$.
    \item The transition $(a_i',v_i',a_{i+1})$ is modified as
    $(a_i',(v_i',-(n-i+m) (v_i'[1] + v_i'[2])), a_{i+1})$.
    \item The transition $(b_i,w_i,b_i')$ is modified as 
    $(b_i,(w_i, -(m-i) (w_i[1] + w_i[2]) - w_i[1]),b_i')$.
    \item The transition $(b_i',w_i',b_{i+1})$ is modified as 
    $(b_i',(w_i', -(m-i) (w_i'[1] + w_i'[2])),b_{i+1})$.
\end{itemize}

Recall that $L = n+(m+1) + (n \cdot (n+m) \cdot 3K) + (n \cdot (n+m) \cdot 2K) + ((m+1) \cdot m \cdot 3K) + (m \cdot m \cdot 2K)$ is chosen so that $C(P)$ for any prefix $P$ of a run always stays in the range $[-L,L]$. We are now ready to prove the following lemma.

\controllingcounter*

\begin{proof}
	By construction of $\mach'$ there is a $\qnz$-run in $\overline{\mach}$ of the form 
	\begin{multline*}
		a_0(1,0) := a_0(g_0,h_0) \actqnz{\alpha_0 v_0} a_0'(g_0',h_0') \actqnz{\alpha_0' v_0'} a_1(g_1,h_1)  \cdots \actqnz{\alpha_{n-1}' v_{n-1}'}  a_n(g_n,h_n) := \\ b_0(k_0,r_0)   \actqnz{\beta_0 w_0} b_0'(k_0',r_0') \actqnz{\beta_0' w_0'} b_1(k_1,r_1) \dots \actqnz{\beta_m w_m} b_m(k_m,r_m) \actqnz{\beta_m' w_m'} b_m' (k_m',r_m') 
	\end{multline*}
	with $k_m' \ge 0, r_m' \ge 1$ iff there is a run in $\mach'$ of the form
	\begin{multline*}
		a_0\big(1,0,L-(n+m+1)\big) := a_0\big(g_0,h_0,L-(n+m+1)\big) \actqnz{\alpha_0 v_0} a_0'\big(g_0',h_0', L-C(P_1')\big)\\ \actqnz{\alpha_0' v_0'} a_1\big(g_1,h_1,L-C(P_1)\big)  \cdots \actqnz{\alpha_{n-1}' v_{n-1}'}  a_n\big(g_n,h_n,L-C(P_n)\big) := b_0\big(k_0,r_0,L-C(Q_0)\big)  \\  \actqnz{\beta_0 w_0} b_0'\big(k_0',r_0',L-C(Q_0')\big) \actqnz{\beta_0' w_0'} b_1\big(k_1,r_1,L-C(Q_1)\big) \cdots \\ \actqnz{\beta_m w_m} b_m\big(k_m,r_m,L-C(Q_m)\big) \actqnz{\beta_m' w_m'} b_m' \big(k_m',r_m',L-C(Q_m')\big) 
	\end{multline*}
	
	where $P_i, P_i', Q_i, Q_i'$ represent the prefix of the run in $\overline{\mach}$ till
	the steps corresponding to the updates $\alpha_i v_i, \alpha_i' v_i', \beta_i w_i, \beta_i w_i'$
	respectively.
	By definition of $L$, $L-C(P_i), L-C(P_i'), L-C(Q_i), L-C(Q_i')$ are all non-negative and 
	so the run in $\mach'$ is indeed a $\qnz$-run.
	Furthermore, since $C(Q_m')$ is the tester of the original $\qnz$-run in $\overline{\mach}$,
	$C(Q_m')$ is always non-negative and furthermore $C(Q_m') = 0$
	iff the final value of the third counter in the $\qnz$-run in $\mach'$ is at least $L$.
	
	Hence, we can conclude that there is a $\qnz$-run 
	from $a_0(1,0)$ covering $b_m'(0,1)$ in $\overline{\mach}$ whose tester equals 0 iff there is a $\qnz$-run from $a_0\big(1,0,L-(n+m+1)\big)$ covering $b_m'(0,1,L)$ in $\mach'$.
\end{proof}

\end{document}